\documentclass{article} 

\usepackage{latexsym,amsfonts,amsthm,amssymb,exscale,enumerate}
\usepackage[centertags,sumlimits,intlimits,namelimits,reqno]{amsmath}
\usepackage{hyperref}
\hypersetup{final=true}

\usepackage{pstricks}
\usepackage{pst-grad}

\psset{linewidth=0.3pt,dimen=middle}
\psset{xunit=.70cm,yunit=0.70cm}

\usepackage{graphicx}
\usepackage[all, knot]{xy}
\SelectTips{cm}{} \xyoption{arc}
\usepackage{epsfig}

        \newcommand\Z{{\mathbb Z}}   
       	\newcommand\R{{\mathbb R}}   
	\newcommand\C{{\mathbb C}}   
        \renewcommand\H{{\mathbb H}}   
        \newcommand\K{{\mathbb K}} 
	\renewcommand\O{{\mathbb O}}   
	\renewcommand\S{{\mathbb S}}

        \newcommand{\OO}{{\rm O}}   
	\newcommand{\SO}{{\rm SO}}

	\newcommand{\SU}{{\rm SU}}   
        
        \newcommand{\Spin}{{\rm Spin}}   
           
	\newcommand{\U}{{\rm U}}

        \newcommand{\h}{{\mathfrak {h}}}   
	\newcommand{\su}{{\mathfrak {su}}}

	\newcommand{\tensor}{\otimes}

        \newcommand{\iso}{\cong}   
           
	\newcommand{\be}{\begin{equation}}   
        \newcommand{\ee}{\end{equation}}   
        \newcommand{\ba}{\begin{eqnarray}}   
        \newcommand{\ea}{\end{eqnarray}}   
        \newcommand{\ban}{\begin{eqnarray*}}   
        \newcommand{\ean}{\end{eqnarray*}}

	\newcommand{\Hilb}{{\rm Hilb}}   
	\newcommand{\Rep}{{\rm Rep}}   
        \newcommand{\maps}{\colon}   
	\newcommand{\tr}{{\rm tr}}

	\newtheorem{theorem}{Theorem}

\begin{document}   
 
	\begin{center}   
	{\bf Division Algebras and Quantum Theory \\}   
	{\em John\ C.\ Baez\\}   
	\vspace{0.3cm}   
	{\small Centre for Quantum Technologies  \\
        National University of Singapore \\
        Singapore 117543  \\  
        \vspace{0.3cm}
        and \\ 
        \vspace{0.3cm}
        Department of Mathematics \\
        University of California \\   
        Riverside CA 92521\\  \quad \\}   
        {\small email:  baez@math.ucr.edu\\} 
	\vspace{0.3cm}   
	{\small April 18, 2011}
	\vspace{0.3cm}   
	\end{center}   

\begin{abstract}
Quantum theory may be formulated using Hilbert spaces over
any of the three associative normed division algebras: 
the real numbers, the complex numbers and the quaternions.
Indeed, these three choices appear naturally in a number of
axiomatic approaches.  However, there are internal problems 
with real or quaternionic quantum theory.  Here we argue that 
these problems can be resolved if we treat real, complex and 
quaternionic quantum theory as part of a unified structure.
Dyson called this structure the `three-fold way'.  It is perhaps
easiest to see it in the study of irreducible unitary representations 
of groups on complex Hilbert spaces.  These representations 
come in three kinds: those that are not isomorphic to their own 
dual (the truly `complex' representations), those 
that are self-dual thanks to a symmetric bilinear pairing 
(which are `real', in that they are the complexifications
of representations on real Hilbert spaces), and those that
are self-dual thanks to an antisymmetric bilinear pairing
(which are `quaternionic', in that they are the underlying
complex representations of representations on quaternionic
Hilbert spaces).  This three-fold classification sheds 
light on the physics of time reversal symmetry, and it already plays
an important role in particle physics.  More generally, 
Hilbert spaces of any one of the three kinds---real, complex and 
quaternionic---can be seen as Hilbert spaces of the other
kinds, equipped with extra structure. 
\end{abstract}

\section{Introduction}

Ever since the birth of quantum mechanics, there has been curiosity about
the special role that the complex numbers play in this theory.  Of
course, one could also ask about the role of real numbers in classical
mechanics: to a large extent, this question concerns the role of the
continuum in physical theories.  In quantum mechanics there is a new 
twist.  Classically, observables are described by real-valued functions.  
In quantum mechanics, they are described by self-adjoint operators on a
complex Hilbert space.  In both theories, the expectation value of an
observable in any state is {\sl real}.  But the complex numbers play a
fundamental role in quantum mechanics that is not apparent in
classical mechanics.  Why?

The puzzle is heightened by the fact that we can formulate many
aspects of quantum theory using real numbers or quaternions instead
of complex numbers.  There are precisely four `normed division
algebras': the real numbers $\R$, the complex numbers $\C$, the
quaternions $\H$ and the octonions $\O$.  Roughly speaking, these are
the number systems extending the reals that have an `absolute value'
obeying the equation $|xy| = |x| \, |y|$.  Since the octonions are
nonassociative it proves difficult to formulate quantum theory
based on these, except in a few special cases.  But the other three
number systems support the machinery of quantum theory quite
nicely.

Indeed, as far back as 1934, Jordan, von Neumann and Wigner~\cite{JNW}
saw real, complex and quaternionic quantum theory appear on an
equal footing in their classification of algebras of observables.
These authors classified {\it finite-dimensional} algebras of 
observables, but later their work was generalized to the 
infinite-dimensional case.  We review some of these classification theorems 
in Section~\ref{classification}.  

At first glance, the lesson seems to be that real, complex and
quaternionic quantum theory stand on an equal footing, at least 
mathematically.  In particular, while quaternionic Hilbert
spaces may be unfamiliar, we shall see their definition soon, and 
they turn out to be quite well-behaved in many respects.  But
experiments seem to show that our universe is described by the complex 
version of quantum theory, not the real or quaternionic version.  
So how did nature decide among these three choices?

In fact, the real and quaternionic versions of quantum theory have
some `problems', or at least striking differences from the familiar
complex version.  First, they lack the usual correspondence between
strongly continuous one-parameter unitary groups and self-adjoint
operators, which in the complex case goes by the name of Stone's
Theorem.  Second, the tensor product of two quaternionic Hilbert
spaces is not a quaternionic Hilbert space.  We discuss these issues 
in Section~\ref{problems}.

One could take these `problems' as the explanation of why nature uses
complex quantum theory, and let the matter drop there.  But this would
be shortsighted.  Indeed, the main claim of this paper is that
instead of being distinct alternatives, real, complex and quaternionic
quantum mechanics are three aspects of a single unified structure.  
Nature did not choose one: she chose all three!

The evidence for this goes back to a very old idea in group theory:
the Frobenius--Schur indicator \cite{FS,Bourbaki2}.  This is a way
of computing a number from an irreducible unitary representation of 
a compact group $G$ on a complex Hilbert space $H$, say 
\[                 \rho \maps G \to \U(H)   \]
where $\U(H)$ is the group of unitary operators on $H$.  
Any such representation is finite-dimensional, so we can take the
trace of the operator $\rho(g^2)$ for any group element $g \in G$, 
and then perform the integral
\[                 \int_G \tr(\rho(g^2)) \, dg \]
where $dg$ is the normalized Haar measure on $G$.
This integral is the {\bf Frobenius--Schur indicator}.  It always 
equals $1$, $0$ or $-1$, and these three cases correspond to whether 
the representation $\rho$ is `real', `complex' or `quaternionic'.  
Here we are using these words in a technical sense.  Remember, 
the Hilbert space $H$ is complex.  However:
\begin{itemize}
\item If the Frobenius--Schur indicator is $1$, then 
$\rho$ is the complexification of a
representation of $G$ on some real Hilbert space.  In this case we
call $\rho$ {\bf real}.
\item If the Frobenius--Schur indicator is $-1$, then 
$\rho$ is the underlying complex
representation of a representation of $G$ on some quaternionic Hilbert
space.  In this case we call $\rho$ {\bf quaternionic}.
\item If the Frobenius--Schur indicator is $0$, then neither of the
above alternatives hold.  In this case we call $\rho$ {\bf complex}.
\end{itemize}

The Frobenius--Schur indicator is an appealingly concrete way to 
decide whether an irreducible representation is real, quaternionic or
complex.  However, the significance of this threefold classification
is made much clearer by another result, which focuses on the key concept
of {\it duality}.  The represention $\rho$ has a dual $\rho^*$, 
which is a unitary representation of $G$ 
on the dual Hilbert space $H^*$.  Of course, the Hilbert space
$H$ has the same dimension as its dual, so they are isomorphic.  
The question thus arises: is the representation $\rho$ isomorphic 
to its dual?  Instead of a simple yes-or-no answer, it turns out 
there are three cases:

\begin{itemize}
\item The representation $\rho$ is real iff it is 
isomorphic to its dual thanks
to the existence of an invariant nondegenerate {\em symmetric} bilinear
form $g \maps H \times H \to \C$.  
\item The representation $\rho$ is quaternionic iff it is
isomorphic to its dual thanks
to the existence of an invariant nondegenerate {\em antisymmetric} bilinear 
form $g \maps H \times H \to \C$.  
\item The representation $\rho$ is complex iff it 
is not isomorphic to its dual.  
\end{itemize}
We recall the proof of this result in Section~\ref{three-fold}.
Freeman Dyson called it the `three-fold way'~\cite{Dyson}.  It also
appears among Vladimir Arnold's list of `trinities'~\cite{Arnold}.
So, in some sense it is well known.  However, its implications for the
foundations of quantum theory seem to have been insufficiently
explored.

What are the implications of the three-fold way?

First, since elementary particles are
often described using irreducible unitary representations of compact
groups, it means that particles come in three kinds: real, complex
and quaternionic!  Of course the details depend not just on the
particle itself, but on the group of symmetries we consider.

For example, take any spin-$\frac{1}{2}$ particle, and consider only
its rotational symmetries.  Then we can describe this particle using a
unitary representation of $\SU(2)$ on $\C^2$.  This representation is
quaternionic.  Why?  Because we can think of a pair of complex numbers
as a quaternion, and $\SU(2)$ as the group of unit quaternions.  
The spin-$\frac{1}{2}$ representation of $\SU(2)$ is then revealed 
to be the action of unit quaternions on the quaternions via left
multiplication.  

In slogan form: {\em qubits are not just quantum---they are also
quaternionic!}  More generally, all particles of half-integer spin are
quaternionic, while particles of integer spin are real, as long as we
consider them only as representations of $\SU(2)$.   

This explains why the square of time reversal is $1$ for 
particles of integer spin, but $-1$ for particles of 
half-integer spin.  Time reversal is `antilinear': it commutes 
with real scalars but anticommutes with multiplication by $i$.   An 
antilinear operator whose square is $1$ acts like complex conjugation.
So, we automatically get such an operator on the complexification of 
any real Hilbert space.   An antilinear operator whose square is $-1$ 
acts like the quaternion unit $j$, since $ij = -ji$ and $j^2 = -1$.  
So, we get such an operator on the underlying complex Hilbert space 
of any quaternionic Hilbert space.  

There is a nice heuristic argument for the same fact using 
`string diagrams', which are a mathematical generalization of 
Feynman diagrams~\cite{BL,BS}.  Any irreducible representation of 
$\SU(2)$ has an invariant nondegenerate bilinear form
$g \maps H \times H \to C$.   This is symmetric:
\[      g(v,w) = g(w,v)   \]
when the spin is an integer, and antisymmetric
\[      g(v,w) = -g(w,v)  .\]
when the spin is a half-integer.  So, in terms of string diagrams:
\[
 \xy  0 ;/r.15pc/:
    (-5,8)*{}="x1";
    (5,8)*{}="x2";
    (-5,5)*{}="m1";
    (5,5)*{}="m2";
    (-5,-5)*{}="k1";
    (5,-5)*{}="k2";
    (-5,-8)*{}="y1";
    (5,-8)*{}="y2";
 \vtwist~{"m1"}{"m2"}{"k1"}{"k2"};
 "x1";"m1" **\dir{-}; 
 "x2";"m2" **\dir{-}; 
 "k1";"y1" **\dir{-}; 
 "k2";"y2" **\dir{-}; 
    "y1";"y2" **\crv{(-5,-14) & (5,-14)};
\endxy
\quad
= 
\quad \pm \; \;
 \xy 0 ;/r.15pc/:
    (-5,8)*{}="x1";
    (5,8)*{}="x2";
    (-5,5)*{}="m1";
    (5,5)*{}="m2";
    (-5,-5)*{}="k1";
    (5,-5)*{}="k2";
    (-5,-8)*{}="y1";
    (5,-8)*{}="y2";
 "m1";"k1" **\dir{-}; 
 "m2";"k2" **\dir{-}; 
 "x1";"m1" **\dir{-}; 
 "x2";"m2" **\dir{-}; 
 "k1";"y1" **\dir{-}; 
 "k2";"y2" **\dir{-}; 
    "y1";"y2" **\crv{(-5,-14) & (5,-14)};
\endxy
\]
If we manipulate this equation a bit, we get:
\[ \xy 0 ;/r.20pc/:
  (0,15)*{}="T";
  (0,-15)*{}="B";
  (0,7.5)*{}="T'";
  (0,-7.5)*{}="B'";
  "T";"T'" **\dir{-}?(.5)*\dir{-};
    "B";"B'" **\dir{-};
    (-4.5,0)*{}="MB";
    (-10.5,0)*{}="LB";
    "T'";"LB" **\crv{(-1.5,-6) & (-10.5,-6)}; \POS?(.25)*+{\hole}="2z";
    "LB"; "2z" **\crv{(-12,9) & (-3,9)};
    "2z";"B'" **\crv{(0,-4.5)};
    (2,13)*{};
    \endxy
    \qquad = \quad \; \pm \; \;
  \xy 0 ;/r.20pc/:
  (0,15)*{}="T";
  (0,-15)*{}="B";
  (0,7.5)*{}="T'";
  (0,-7.5)*{}="B'";
  "T";"B" **\dir{-}?(.5)*\dir{-};
    (2,13)*{};
    \endxy
 \]
The diagram at left shows a particle reversing its direction in time
twice, so not surprisingly, the square of time reversal is $\pm 1$
depending on whether the spin is an integer or half-integer.
But we can also see more.  Suppose we draw the diagram at left as a 
`ribbon':
\[
 \xy
 (0,0)*{\includegraphics[angle=180,scale=0.5]{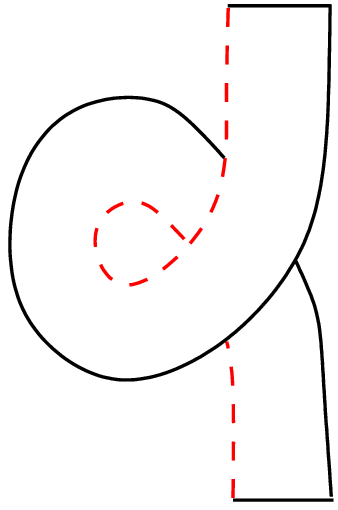}};
 \endxy
\]
Then, pulling it tight, we get a ribbon with a 360 degree twist in it:
\[
 \xy
 (0,0)*{\includegraphics[angle=180,scale=0.5]{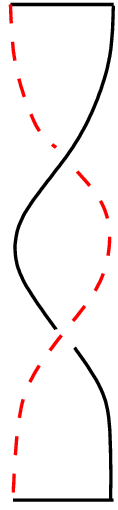}};
 \endxy
\]
And indeed, rotating a particle by 360 degrees gives a phase of $+1$ 
if its spin is an integer, and $-1$ if its spin is a half-integer.  
So this well-known fact is yet another manifestation of the 
real/quaternionic distinction.  We explain these ideas further
in Section~\ref{physics}.  

More broadly, the three-fold way suggests that complex 
quantum theory {\em contains} the real and quaternionic theories.
Indeed, this can be made precise with the help of some
category theory.  Let $\Hilb_\R$, $\Hilb_\C$ and $\Hilb_\H$
stand for the categories of real, complex and quaternionic Hilbert
spaces, respectively.  Then there are `faithful functors' from
$\Hilb_\R$ and $\Hilb_\H$ into $\Hilb_\C$.  This is a way of making 
precise the fact that we can treat real or quaternionic Hilbert 
spaces as complex Hilbert spaces equipped with extra structure.   
The extra structure is just a nondegenerate bilinear form 
$g \maps H \times H \to \C$, which is symmetric in the real case 
and antisymmetric in the quaternionic case.

Since complex quantum theory contains the real and quaternionic
theories, we might be tempted to conclude that the complex theory is
fundamental, with the other theories arising as offshoots.  But 
we can also faithfully map $\Hilb_\C$ and $\Hilb_\H$ into $\Hilb_\R$.  
We can even faithfully map $\Hilb_\R$ and $\Hilb_\C$ into $\Hilb_\H$!  We
describe these faithful functors in Section~\ref{categories}.  So, 
one may argue that no one form of quantum mechanics is primordial: each 
contains the other two!  We provide details in Section~\ref{categories}.  
Then, in Section~\ref{solutions}, we use the relation between the three
forms of quantum theory to solve the problems with real and quaternionic
quantum mechanics raised earlier.

\section{Classifications}
\label{classification}

The laws of algebra constrain the possibilities for theories that
closely resemble the quantum mechanics we know and love.  Various
classification theorems delimit the possibilities.  Of course,
theorems have hypotheses.  It is easy to get around these theorems by
weakening our criteria for what counts as a theory `closely
resembling' quantum mechanics; if we do this, we can find a large
number of alternative theories.  This is especially clear in the
category-theoretic framework, where many theories based on convex sets
of states give rise to categories of physical systems and processes
sharing some features of standard quantum mechanics~\cite{Barnum1,
Barnum2,Barnum3}.  Here, however, we focus on results 
that pick out real, complex and quaternionic quantum mechanics as special.  

Our treatment is far from complete; it is merely meant as a sketch of
the subject.   In particular, while we describe the Jordan algebra approach
to quantum systems with finitely many degrees of freedom, and also the
approach based on convex sets, we neglect the approach based on lattices
of propositions.  Good introductions to this enormous subject include
the paper by Steirteghem and Stubbe~\cite{SS} and the classic texts by
Piron~\cite{Piron} and Varadarajan~\cite{Varadarajan}.

\subsection{Normed division algebras}
\label{division}

After discovering that complex numbers could be viewed as simply pairs
of real numbers, Hamilton sought an algebra of `triples' where
addition, subtraction, multiplication and division obeyed most of the
same rules.  Alas, what he was seeking did not exist.  After much
struggle, he discovered the `quaternions': an algebra consisting of
expressions of the form $a1 + bi + cj + dk$ ($a,b,c,d \in \R$),
equipped with the associative product with unit $1$ uniquely
characterized by these equations:
\[                i^2 = j^2 = k^2 = ijk = -1.   \]
He carved these equations into a bridge as soon he discovered them.
He spent the rest of his life working on the
quaternions, so this algebra is now called $\H$ in his honor.  

The day after Hamilton discovered the quaternions, he sent a letter 
describing them to his college friend John Graves.  A few months later
Graves invented yet another algebra, now called the `octonions'
and denoted $\O$.  The octonions are expressions of the form 
\[          a_0 1 + a_1 e_1 + a_2 e_2 + \cdots + a_7 e_7  \]
with $a_i \in \R$.  The multiplication of octonions is not 
associative, but it is still easily described using the {\bf Fano
plane}, a projective plane with 7 points and 7 lines:

\centerline{\epsfysize=1.5in\epsfbox{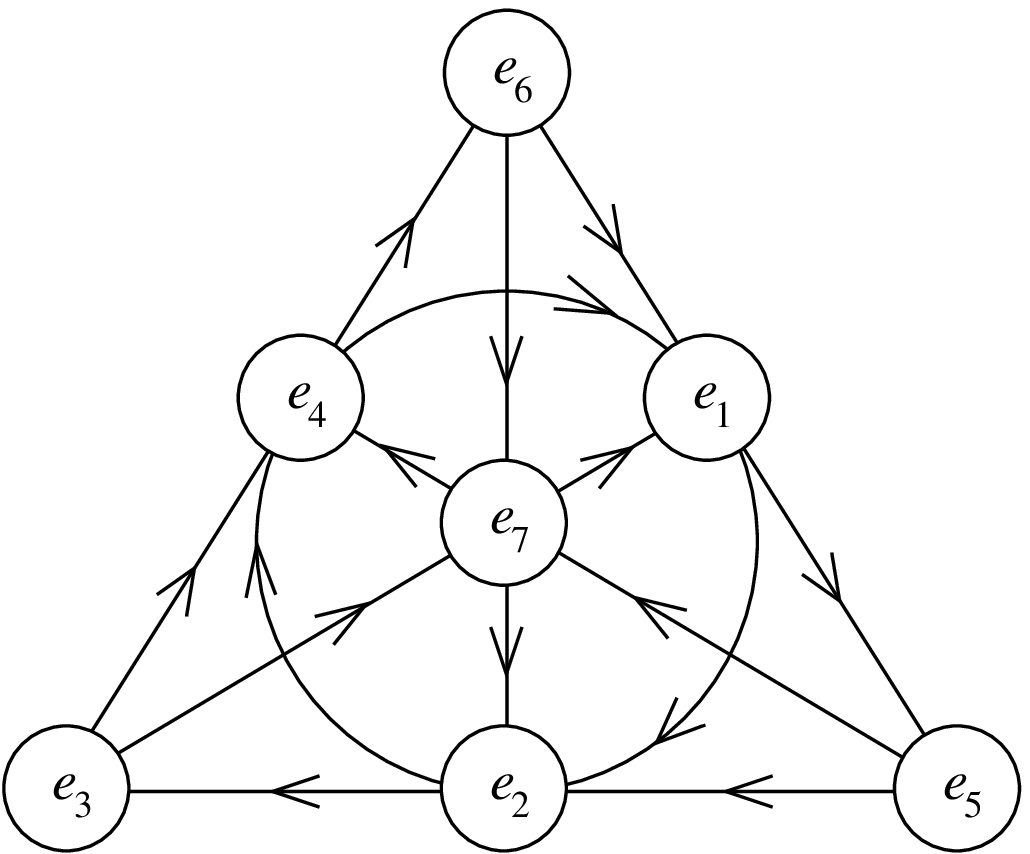}}
\label{Fano}
\medskip

The `lines' here are the sides of the triangle, its altitudes, and the
circle containing the midpoints of the sides.  Each line contains
three points, and each of these triples has a cyclic ordering, as
shown by the arrows.  If $e_i, e_j,$ and $e_k$ are cyclically ordered
in this way then
\[            e_i e_j = e_k,  \qquad e_j e_i = -e_k  . \]
Together with these rules:
\begin{itemize}
\item $1$ is the multiplicative identity,
\item $e_1, \dots, e_7$ are square roots of -1,
\end{itemize}
the Fano plane completely describes the algebra structure of the
octonions.  As an exercise, we urge the reader to check that the
octonions are nonassociative.

What is so great about the number systems discovered by Hamilton and
Graves?  Like the real and complex numbers, they are `normed division
algebras'.  For us, an {\bf algebra} will be a real vector space 
$A$ equipped with a multiplication
\[                     
\begin{array}{ccl}
A \times A & \to & A \\
(x, y)     &\mapsto& xy
\end{array}
\]
that is real-linear in each argument.  We do not assume our
algebras are associative.  We say an algebra $A$ is {\bf 
unital} if there exists an element $1 \in A$ with
\[               1 a = a = a 1 \]
for all $a \in A$.  And we define a
{\bf normed division algebra} to be a 
unital algebra equipped with a norm 
\[                | \cdot | \maps A \to [0, \infty) \]
obeying the rule
\[                |ab| = |a| \, |b|  .\]
As a consequence, we have 
\[     ab = 0 \quad \implies \quad a = 0 \mathrm{ \; or \; } b = 0 .\]

The real and complex numbers are obviously normed division algebras.
For the quaternions we can define the norm to be:
\[          |a + bi + cj + dk| = \sqrt{a^2 + b^2 + c^2 + d^2} . \]
A similar formula works for the octonions:
\[      |a_0 + a_1 e_1 + \cdots + a_7 e_7| = \sqrt{a_0^2 + \cdots + a_7^2}. \]
With some sweat, one can check that these rules make $\H$ and $\O$
into normed division algebras.  

The marvelous fact is that there are no more!  In an 1898 paper,
Hurwitz proved that $\R$, $\C$, $\H$ and $\O$ are the only
finite-dimensional normed division algebras~\cite{Hurwitz}.  In 1960,
Urbanik and Wright~\cite{UW} removed the finite-dimensionality
condition:

\begin{theorem}\label{Urbanik and Wright} Every normed division
algebra is isomorphic to either $\R$, $\C$, $\H$ or $\O$.
\end{theorem}

For quantum mechanics, it is important that every normed division
algebra is a {\bf{\boldmath{$\ast$}}-algebra}, meaning it is equipped
with a real-linear map
\[                     
\begin{array}{ccl}
A & \to & A \\
x   &\mapsto& x^*
\end{array}
\]
obeying these rules:
\[   (xy)^* = y^* x^* , \qquad           (x^*)^* = x. \]
For $\C$ this map is just complex conjugation, while for
$\R$ it is the identity map.   For the quaternions it is given by:
\[      (a1 + bi + cj + dk)^* = a - bi - cj - dk \]
and for the octonions,
\[      (a_0 1 + a_1 e_1 + \cdots + a_7 e_7)^* = 
 a_0 1 - a_1 e_1 - \cdots - a_7 e_7  .\]
One can check in all four cases we have
\[              xx^* = x^* x = |x|^2 1 .\]

For the three associative normed division algebras, the $\ast$-algebra
structure lets us set up a theory of Hilbert spaces.  Let us quickly
sketch how.  Suppose $\K$ is an associative normed division algebra.
Then we define a {\bf \boldmath $\K$-vector space} to be a right
$\K$-module: that is, an abelian group $V$ equipped with a map
\[                     
\begin{array}{ccl}
V \times \K & \to & V \\
(v, x)     &\mapsto& vx
\end{array}
\]
that obeys the laws
\[           (v + w)(x) = vx + wx, \qquad v(x + y) = vx + vy, \]
\[         (vx)y = v(xy) .\]
We say a map $T \maps V \to V'$ between $\K$-vector spaces 
is {\bf \boldmath $\K$-linear} if 
\[                T(vx + wy) = T(v)x + T(w)y  \]
for all $v,w \in V$ and $x,y \in \K$.   When no confusion can arise,
we call a $\K$-linear map $T \maps V \to V'$ a {\bf linear operator}
or simply an {\bf operator}.

The reader may be appalled that we are multiplying by scalars on the
{\it right} here.  It makes no difference except for the quaternions,
which are noncommutative.  Even in that case, we could use 
either left or right multiplication by scalars, as long as we stick to
one convention.  But since we write the operator $T$ on the left in 
the expression $T(v)$, it makes more sense to do scalar
multiplication on the right, so no symbols trade places 
in the law $T(vx) = T(v)x$.

There is a category of $\K$-vector spaces and operators between
them.  Even in the quaternionic case every $\K$-vector space has 
a basis, and any two bases have the same
cardinality, so we can talk about the dimension of a vector space over
$\K$, and every finite-dimensional vector space over $\K$ is
isomorphic to $\K^n$.  Every operator $T \maps \K^n \to \K^m$ can be
written as matrix:
\[                (Tv)_i = \sum_j T_{ij} v_j . \]

The $\ast$-algebra structure of the normed division algebras
becomes important when we study Hilbert spaces.
We define an {\bf inner product} on a $\K$-vector space $V$ to
be a map
\[
\begin{array}{ccl}
V \times V & \to & \K \\
(v,w)   &\mapsto& \langle v,w \rangle
\end{array}
\]
that is $\K$-linear in the second argument, has
\[      \langle v, w \rangle = \langle w, v \rangle^* \]
for all $v, w \in V$, and is {\bf positive definite}:
\[          \langle v, v \rangle \ge 0 \]
with equality only if $v = 0$.
As usual, an inner product gives a norm:
\[            \|v\| = \sqrt{\langle v, v\rangle} , \]
and we say a $\K$-vector space with inner product is a {\bf \boldmath
$\K$-Hilbert space} if this norm makes the vector space into a complete
metric space.  Every finite-dimensional $\K$-Hilbert space is isomorphic
to $\K^n$ with its standard inner product
\[             \langle v, w \rangle = \sum_{i = 1}^n v_i^* w_i .\]

This is fine for $\R$, $\C$ and $\H$: what about $\O$?  It would be nice
if we could think of $\O^n$ as an $n$-dimensional octonionic vector space,
with scalar multiplication defined in the obvious way.
Unfortunately, the law
\[          (vx)y = v(xy)  \]
fails, because the octonions are nonassociative!  Furthermore, there are
no maps $T \maps \O^n \to \O^m$ obeying
\[          T(vx) = T(v)x  \]
except the zero map.   So, nobody has managed to develop a good theory
of octonionic linear algebra.
 
While the octonions are nonassociative, they are still {\bf alternative}.
That is, the associative law holds whenever two of the elements being
multiplied are equal:
\[    (xx)y = x(xy) \qquad (xy)x = x(yx) \qquad (yx)x = y(xx)  .\]
Surprisingly, this is enough to let us carry out a bit of quantum 
theory as if $\O$, $\O^2$ and $\O^3$ were well-defined 
octonionic Hilbert spaces.  The 3-dimensional case is by far the most
exciting: it leads to a structure called the `exceptional Jordan 
algebra'.  But beyond dimension 3, it seems there is little to say.

\subsection{Jordan algebras}
\label{jordan}

In 1932, Pascual Jordan tried to isolate some axioms that an `algebra
of observables' should satisfy~\cite{Jordan}.  The unadorned phrase
`algebra' usually signals an associative algebra, but this is not the
kind of algebra Jordan was led to.  In both classical and quantum
mechanics, observables are closed under addition and multiplication by
real scalars.  In classical mechanics we can also multiply
observables, but in quantum mechanics this becomes problematic.  After
all, given two bounded self-adjoint operators on a complex Hilbert
space, their product is self-adjoint if and only if they commute.  

However, in quantum mechanics one can still raise an observable
to a power and obtain another observable.  From squaring and
taking real linear combinations, one can construct a commutative product:
\[     a \circ b = \frac{1}{2}((a+b)^2 - a^2 - b^2)
                  = \frac{1}{2}(ab + ba) . \] 
This product is not associative, but it is {\bf power-associative}: any way 
of parenthesizing a product of copies of the same observable $a$ gives
the same result.  This led Jordan to define what is now called a {\bf
formally real Jordan algebra}: a real vector space with a bilinear,
commutative and power-associative product satisfying
\[  a_1^2 + \cdots + a_n^2 = 0 \quad \implies \quad a_1 = \cdots = a_n = 0  \]
for all $n$.  It turns out that every finite-dimensional formally real 
Jordan algebra is in fact unital.  Moreover, the last condition gives 
$A$ a partial ordering: if we write $a \le b$ when the element 
$b - a$ is a sum of squares, it says 
\[   a \le b \textrm{\; and \; } b \le a \; \quad \implies \quad a = b .  \]
So, in a formally real Jordan algebra it makes sense to speak of one
observable being `greater' than another.

In 1934, Jordan published a paper with von Neumann and Wigner
classifying finite-dimensional formally real Jordan algebras~\cite{JNW}.  
They began by proving that any such algebra is a direct
sum of `simple' ones.  A formally real Jordan algebra is {\bf simple}
when it has exactly two different ideals, $\{0\}$ and $A$ itself.
Here, as usual, an {\bf ideal} is a vector subspace $B \subseteq A$ 
such that $b \in B$ implies $a \circ b \in B$ for all $a \in A$.  

And then, they proved:
\begin{theorem}  Every simple finite-dimensional formally real Jordan
algebra is isomorphic to exactly one on this list:
\begin{itemize}
\item The algebras $\h_n(\R)$ of $n \times n$ self-adjoint real 
matrices with product $a \circ b = \frac{1}{2}(ab + ba)$, where $n \ge 3$.
\item The algebras $\h_n(\C)$ of $n \times n$ self-adjoint complex 
matrices with product $a \circ b = \frac{1}{2}(ab + ba)$, where $n \ge 3$.
\item The algebras $\h_n(\H)$ of $n \times n$ self-adjoint quaternionic 
matrices with product $a \circ b = \frac{1}{2}(ab + ba)$, where $n \ge 3$.
\item The algebra $\h_3(\O)$ of $3 \times 3$ self-adjoint octonionic 
matrices with the product $a \circ b = \frac{1}{2}(ab + ba)$.
\item The {\bf spin factors} for $n \ge 0$, namely the 
vector spaces $\R^n \oplus \R$ with product
\[  (x,t) \circ (x', t') =
(t x' + t' x, x \cdot x' + tt').  \]
\end{itemize}
\end{theorem}
\noindent
Here we say a square matrix $T$ is {\bf self-adjoint} if $T_{ji} = 
(T_{ij})^*$.  For $\K = \R,\C,$ or $\H$, we can identify a 
self-adjoint $n \times n$ matrix with an operator
$T \maps \K^n \to \K^n$ that is {\bf self-adjoint} in the sense
that 
\[         \langle T v, w \rangle = \langle v, T w \rangle  \]
for all $v,w \in \K^n$.  In the octonionic case we do not know
what Hilbert spaces and operators are, but we can still work 
with matrices.  

The $1 \times 1$ self-adjoint matrices with entries in any normed
division algebra form a Jordan algebra isomorphic to the real numbers.
This is also isomorphic to the spin factor with $n = 0$.  
As we shall see very soon, 
the $2 \times 2$ self-adjoint matrices with entries in any normed 
division algebra also form a Jordan algebra isomorphic to a spin factor.  
Curiously, in the octonionic case we cannot go beyond $3 \times 3$ 
self-adjoint matrices and still get a Jordan algebra.  The 
$3 \times 3$ self-adjoint octonionic matrices form a 27-dimensional 
formally real Jordan algebra called the {\bf exceptional Jordan algebra}.

What does all this mean for physics?  The spin factors are related
to Clifford algebras and spinors, hence their name.  They also have an
intriguing relation to special relativity, since $\R^n \oplus \R$ can
be identified with $(n+1)$-dimensional Minkowski spacetime, and its
cone of positive elements is then revealed to be none other than 
the future lightcone.  Furthermore, we have some interesting coincidences:
\begin{itemize}
\item  
The Jordan algebra $\h_2(\R)$ is
isomorphic to the spin factor $\R^2 \oplus \R$.
\item
The Jordan algebra $\h_2(\C)$ is
isomorphic to the spin factor $\R^3 \oplus \R$.
\item 
The Jordan algebra $\h_2(\H)$ is 
isomorphic to the spin factor $\R^5 \oplus \R$.
\item
The Jordan algebra $\h_2(\O)$ is
isomorphic to the spin factor $\R^9 \oplus \R$.
\end{itemize}
This sets up a relation between the real numbers, complex numbers, 
quaternions and octonions and the Minkowski spacetimes of dimensions
3,4,6 and 10.  These are precisely the dimensions where a classical
superstring Lagrangian can be written down!  Far from being a 
coincidence, this is the tip of a huge and still not fully fathomed
iceberg, which we have discussed elsewhere~\cite{Octonions,SUSY1,SUSY2}.  

The exceptional Jordan algebra remains mysterious.  Practically ever
since it was discovered, physicists have looked for some application
of this entity.  For example, when it was first found that quarks come
in three colors, Okubo and others hoped that $3 \times 3$ self-adjoint
octonionic matrices might serve as observables for these exotic
degrees of freedom~\cite{Okubo}.  Alas, nothing much came of this.
More recently, people have discovered some relationships between
10-dimensional string theory and the exceptional Jordan algebra,
arising from the fact that the $2 \times 2$ self-adjoint octonionic
matrices can be identified with 10-dimensional Minkowski spacetime~\cite{CH}.  
But nothing definitive has emerged so far.

In 1983, Zelmanov generalized the Jordan--von Neumann--Wigner
classification to the infinite-dimensional case, working with Jordan
algebras that need not be formally real~\cite{Zelmanov}.  In any
formally real Jordan algebra, the following peculiar law holds:
\[           (a^2 \circ b) \circ a = a^2 \circ (b \circ a).  \]
Any vector space with a commutative bilinear product obeying this law
is called a {\bf Jordan algebra}.  Zelmanov classified the simple
Jordan algebras and proved they are all of three kinds: a kind
generalizing Jordan algebras of self-adjoint matrices, a kind
generalizing spin factors, and an `exceptional' kind.  For a good
introduction to this work, see McCrimmon's book~\cite{McCrimmon}.

Zelmanov's classification theorem does not highlight the special role
of the reals, complex numbers and quaternions.  Indeed, when we drop
the `formally real' condition, a host of additional finite-dimensional
simple Jordan algebras appear, beside those in Jordan, von Neumann and
Wigner's classification theorem.  While it is an algebraic {\it tour 
de force}, nobody has yet used Zelmanov's theorem to shed new light 
on quantum theory.  So, we postpone infinite-dimensional considerations 
until Section~\ref{soler}, where we discuss the remarkable theorem of 
Sol\`er.

\subsection{Convex cones}

The formalism of Jordan algebras seems rather removed from the actual
practice of physics, because in quantum theory we hardly ever take two
observables $a$ and $b$ and form their Jordan product $\frac{1}{2}(ab +
ba)$.  As hinted in the previous section, it is better to think of
this operation as derived from the process of {\it squaring} an
observable, which is something we actually do.  But still, we cannot
help wondering: does the classification of finite-dimensional formally
real Jordan algebras, and thus the special role of normed division
algebras, arise from some axiomatic framework more closely tied to
quantum physics as it usually practiced?

One answer involves the correspondence between states and observables.
Consider first the case of ordinary quantum theory.  If a quantum
system has the Hilbert space $\C^n$, observables are described by
self-adjoint $n \times n$ complex matrices: elements of the Jordan
algebra $\h_n(\C)$.  But matrices of this form that are nonnegative
and have trace 1 also play another role.  They are called {\bf density
matrices}, and they describe {\it states} of our quantum system: not
just pure states, but also more general mixed states.  The idea is
that any density matrix $\rho \in \h_n(\C)$ allows us to define
expectation values of observables $a \in \h_n(\C)$ via
\[                   \langle a \rangle = \tr(\rho a) .\]
The map sending observables to their expectation values is
real-linear.  The fact that $\rho$ is nonnegative is equivalent to
\[                 a \ge 0 \; \implies \; \langle a \rangle \ge 0 \]
and the fact that $\rho$ has trace 1 is equivalent to
\[                  \langle 1 \rangle = 1  .\]

Of course, states lie in the dual of the space of observables:
that much is obvious, given that the expectation value of an
observable should depend linearly on the observable.  The more
mysterious thing is that we can identify the vector space of
observables with its dual:
\[ \begin{array}{ccl}
\h_n(\C) &\cong& \h_n(\C)^*  \\
    a &\mapsto& \langle a, \cdot \rangle 
\end{array}
\]
using the trace, which puts a real-valued inner product on the
space of observables:
\[     \langle a, b \rangle = \tr(ab)  .\]
Thus, states also {\it correspond to} certain observables: the
nonnegative ones having trace 1.  We call this fact the {\bf
state-observable correspondence}.

All this generalizes to an arbitrary finite-dimensional formally real
Jordan algebra $A$.  Every such algebra automatically has an identity
element~\cite{JNW}.  This lets us define a {\bf state} on $A$ to be a
linear functional $\langle \cdot \rangle \maps A \to \R$ that is {\bf
nonnegative}:
\[  a \ge 0  \implies \langle a \rangle \ge 0  \]
and {\bf normalized}:
\[   \langle 1 \rangle = 1 .\]
But in fact, there is a one-to-one correspondence between 
linear functionals on $A$ and elements of $A$.  The reason is that
every finite-dimensional Jordan algebra has a {\bf trace}
\[        \tr \maps A \to \R \]
defined so that $\tr(a)$ is the trace of the linear operator
\[            
\begin{array}{ccl}
A & \to & A  \\
b & \mapsto & a \circ b .
\end{array}
\]       
Such a Jordan algebra is then formally real if and only if
\[        \langle a, b \rangle = \tr(a \circ b) \]
is a real-valued inner product.   So, when $A$ is a finite-dimensional
formally real Jordan algebra, any linear functional 
$\langle \cdot \rangle \maps A \to \R$ can be written as
\[          \langle a \rangle = \tr(\rho \circ a) \]
for a unique element $\rho \in A$.  Conversely, every element $\rho
\in A$ gives a linear functional by this formula.  While not obvious,
it is true that the linear functional $\langle \cdot \rangle$ is
nonnegative if and only if $\rho \ge 0$ in terms of the ordering on
$A$.  More obviously, $\langle \cdot \rangle$ is normalized if and
only if $\tr(\rho) = 1$.  So, states can be identified with certain
special observables: namely, those observables $\rho \in A$ with
$\rho \ge 0$ and $\tr(\rho) = 1$.  

In short: whenever the observables in our theory form a
finite-dimensional formally real Jordan algebra, we have a
state-observable correspondence.  But what is the physical meaning of
the state-observable correspondence?  Why in the world should {\it
states} correspond to special {\it observables?}

Here is one attempt at an answer.  Every finite-dimensional formally
real Jordan algebra comes equipped with a distinguished observable,
the most boring one of all: the identity, $1 \in A$.  This is
nonnegative, so if we normalize it, we get an observable
\[             \rho_0 = \frac{1}{\tr(1)} \, 1 \in A \]
of the special kind that corresponds to a state.  This state, say
$\langle \cdot \rangle_0$, is just the normalized trace:
\[           \langle a \rangle_0 = \tr(\rho_0 \circ a) = 
\frac{\tr(a)}{\tr(1)}   .\]
And this state has a clear physical meaning: it is the {\bf state of
maximal ignorance}!  It is the state where we know as little as
possible about our system---or more precisely, at least in the case of
ordinary complex quantum theory, the state where entropy is maximized.

For example, consider $A = \h_2(\C)$, the algebra of observables of
a spin-$\frac{1}{2}$ particle.  Then the space of states is the so-called
`Bloch sphere', really a 3-dimensional ball.  On the surface of this
ball are the pure states, the states where we know as much as possible
about the particle.  At the center of the ball is the state of maximum
ignorance.  This corresponds to the density matrix
\[            \rho_0 = \left( \begin{array}{cc} 
                       \frac{1}{2} & 0 \\
                         0 & \frac{1}{2} 
                       \end{array} \right) .
\]
In this state, when we measure the particle's spin along any axis,
we have a $\frac{1}{2}$ chance of getting spin up and a $\frac{1}{2}$
chance of spin down.

Returning to the general situation, note that $A$ always acts on its
dual $A^*$: given $a \in A$ and a linear functional $\langle \cdot
\rangle$, we get a new linear functional $\langle a \circ \cdot
\rangle$.  This captures the idea, familiar in quantum theory, that
observables are also `operators'.  The state-observable
correspondence arises naturally from this: we can get any state by
acting on the state of complete ignorance with a suitable observable.
After all, any state corresponds to some observable $\rho$, as follows:
\[          \langle a \rangle = \tr(\rho \circ a) \]
So, we can get this state by acting on the state of
maximal ignorance, $\langle \cdot \rangle_0$, by the observable 
$\tr(1) \rho$:
\[         \langle \tr(1)\rho \circ a \rangle_0 = 
\frac{\tr(1)}{\tr(1)} \tr(\rho \circ a) =
\langle a \rangle .\]

So, we see that the state-observable correspondence springs from two
causes.  First, there is a distinguished state, the state of maximal
ignorance.  Second, any other state can be obtained from the state of
maximal ignorance by acting on it with a suitable observable.  

While these thoughts raise a host of questions, they also help
motivate an important theorem of Koecher and Vinberg~\cite{Koecher,
Vinberg}.  The idea is to axiomatize the situation we have just
described, in a way that does not mention the Jordan product in $A$,
but instead emphasizes:
\begin{itemize}
\item the state-observable correspondence, and
\item the fact that `positive' observables, namely those whose
observed values are always positive, form a cone.
\end{itemize}

To find appropriate axioms, suppose $A$ is a finite-dimensional
formally real Jordan algebra.  Then seven facts are always
true.  First, the set of positive observables 
\[              C = \{a \in A \colon a > 0\} \]
is a {\bf cone}: that is, $a \in C$ implies that every positive
multiple of $a$ is also in $C$.  Second, this cone is {\bf convex}:
if $a,b \in C$ then any linear combination $xa + (1-x)b$ with 
$0 \le x \le 1$ also lies in $C$.  Third, it is an open set.  Fourth, it is
{\bf regular}, meaning that if $a$ and $-a$ are both in the closure
$\overline{C}$, then $a = 0$.  This last condition may seem obscure,
but if we note that
\[             \overline{C} = \{ a \in A \colon a \ge 0 \}  \]
we see that $C$ being regular simply means
\[              a \ge 0 \textrm{ \; and \; } -a \ge 0 \quad
\implies \quad a = 0 , \]
a perfectly plausible assumption.

Next recall that $A$ has an inner product; this is what lets us
identify linear functionals on $A$ with elements of $A$.  This also
lets us define the {\bf dual cone}
\[           C^* = \{  a \in A \colon \; \forall b \in C \; 
\langle a,b \rangle > 0 \} \]
which one can check is indeed a cone.  The fifth fact about $C$ is
that it is {\bf self-dual}, meaning $C = C^*$.  This formalizes the
state-observable correspondence, since it means that states correspond
to special observables.  All the elements $a \in C$ are positive
observables, but certain special ones, namely those with 
$\langle a, 1 \rangle = 1$, can also be viewed as states.

The sixth fact is that $C$ is {\bf homogeneous}: given any two points
$a,b \in C$, there is a real-linear transformation $T \maps A
\to A$ mapping $C$ to itself in a one-to-one and onto way, with the
property that $Ta = b$.  This says that cone $C$ is highly
symmetrical: no point of $C$ is any `better' than any other, at least
if we only consider the linear structure of the space $A$, ignoring
the Jordan product and the trace.

From another viewpoint, however, there is a very special point of $C$,
namely the identity $1$ of our Jordan algebra.  And this brings us to
our seventh and final fact: the cone $C$ is {\bf pointed}, meaning
that it is equipped with a distinguished element (in this case $1 \in
C$).  As we have seen, this element corresponds to the `state of
complete ignorance', at least after we normalize it.

In short: when $A$ is a finite-dimensional formally real Jordan
algebra, $C$ is a pointed homogeneous self-dual regular open convex
cone.  

In fact, there is a category of pointed homogeneous self-dual regular
open convex cones, where:
\begin{itemize}
\item  An object is a finite-dimensional real inner product space 
$V$ equipped with a pointed homogeneous self-dual regular open convex 
cone $C \subset V$.
\item   
A morphism from one object, say $(V,C)$, to another, say $(V',C')$,
is a linear map $T \maps V \to V'$ preserving the inner product and
mapping $C$ into $C'$.  
\end{itemize}

Now for the payoff.  The work of Koecher and Vinberg~\cite{Koecher,
Vinberg}, nicely explained in Koecher's 
Minnesota notes~\cite{Koecher2}, shows that:
\begin{theorem} The category of pointed homogeneous self-dual 
regular open convex cones is equivalent to the category of
finite-dimensional formally real Jordan algebras.
\end{theorem}
\noindent
This means that the theorem of Jordan, von Neumann and Wigner,
described in the previous section, also classifies the pointed
homogeneous self-dual regular convex cones!
\begin{theorem} Every pointed homogeneous self-dual 
regular open convex cones is isomorphic to a direct sum of 
those on this list:
\begin{itemize}
\item the cone of positive elements in $\h_n(\R)$ for $n \ge 3$,
\item the cone of positive elements in $\h_n(\C)$ for $n \ge 3$,
\item the cone of positive elements in $\h_n(\H)$ for $n \ge 3$,
\item the cone of positive elements in $\h_3(\O)$, 
\item the future lightcone in $\R^n \oplus \R$ for $n \ge 0$.
\end{itemize}
\end{theorem}
\noindent
Some of this deserves a bit of explanation.  For $\K = \R, \C, \H$, 
an element $T \in \h_n(\K)$ is {\bf positive} if and only if the 
corresponding operator $T \maps \K^n \to \K^n$ has 
\[            \langle v, Tv \rangle > 0  \]
for all nonzero $v \in \K^n$.  A similar trick works for defining
positive elements of $\h_3(\O)$, but we do not need the details here.
We say an element $(x,t) \in \R^n \oplus \R$ lies in the 
{\bf future lightcone} if $t > 0$ and $t^2 - x \cdot x > 0$.  This
of course fits in nicely with the idea that the spin factors are
connected to Minkowski spacetimes.  Finally, there is an obvious
notion of direct sum for Euclidean spaces with cones, where the
direct sum of $(V,C)$ and $(V',C')$ is $V \oplus V'$ equipped with 
the cone
\[             C \oplus C' = \{(v,v') \in V\oplus V' \colon \; 
v \in C, v' \in C' \} .\]

In summary, self-adjoint operators on real, complex and quaternionic
Hilbert spaces arise fairly naturally as observables starting from a
formalism where the observables form a cone, and we insist on an
state-observable correspondence.  There is a well-developed approach
to probabilistic theories that works for cones that are neither
self-dual nor homogeneous: see for example the work of Barnum and
coauthors~\cite{Barnum1,Barnum2}.  This has already allowed these
authors to shed new light on the physical significance of 
self-duality~\cite{Barnum3}.  But perhaps further thought on the 
state-observable
correspondence will clarify the meaning of the Koecher-Vinberg theorem,
and help us better understand the appearance of normed division 
algebras in quantum theory.  

Finally, we should mention that in pondering the state-observable
correspondence, it is worthwhile comparing the `state-operator
correspondence'.  This is best known in the context of string
theory~\cite{Polchinski}, but it really applies whenever we have a
$C^*$-algebra of observables, say $A$, equipped with a state $\langle
\cdot \rangle \maps A \to \C$.  Then the Gelfand-Naimark-Segal
construction lets us build a Hilbert space $H$ on which $A$ acts,
together with a distinguished unit vector $v \in H$ called the `vacuum
state'.  The Hilbert space $H$ is built by taking a completion of a
quotient of $A$, so a dense set of vectors in $H$ come from elements
of $A$.  So, we get states from certain observables.  In particular,
the vacuum state $v$ comes from the element $1 \in A$.

This is reminiscent of how in our setup, the state of maximal
ignorance comes from the element $1$ in the Jordan algebra of
observables.  But there are also some differences: for example, the
Gelfand-Naimark-Segal construction requires choosing a state, and it works 
for infinite-dimensional $C^*$-algebras, while our construction works for
finite-dimensional formally real Jordan algebras, which have a
canonical state: the state of maximum ignorance.  Presumably both
constructions are special cases of some more general construction.

\subsection{Sol\`er's theorem}
\label{soler}

In 1995, Maria Pia Sol\`er~\cite{Soler} proved a result which has
powerful implications for the foundations of quantum mechanics.  The
idea of Sol\`er's theorem is to generalize the concept of Hilbert
space beyond real, complex and quaternionic Hilbert spaces, but then
show that under very mild conditions, the only options for {\it
infinite-dimensional} Hilbert spaces are the usual ones!  Our
discussion here is largely based on the paper by Holland~\cite{Holland}.

To state Sol\`er's theorem, we begin by saying what sort of ring
our generalized Hilbert spaces will be vector spaces over.  
We want rings where we can divide, and we want them to have
an operation like complex conjugation. 
A ring is called a {\bf division ring} if every nonzero element
$x$ has an element $x^{-1}$ for which $x x^{-1} = x^{-1} x = 1$.
It is called a {\bf \boldmath{$\ast$}-ring} if it is equipped
with a function $x \mapsto x^*$ for which:
\[    (x + y)^* = x^* + y^*, \qquad (xy)^* = y^* x^* , \qquad
       (x^*)^* = x.\]
A {\bf division \boldmath{$\ast$}-ring} is a $\ast$-ring that
is also a division ring.  

Now, suppose $\K$ is a division $\ast$-ring.  As before, we define
a {\bf \boldmath $\K$-vector space} to be a right $\K$-module.  And
as before, we say a function $T \maps V \to V'$ between $\K$-vector 
spaces is {\bf \boldmath $\K$-linear} if 
\[                T(vx + wy) = T(v)x + T(w)y  \]
for all $v,w \in V$ and $x,y \in \K$.  
We define a {\bf hermitian form} on a $\K$-vector space $V$ to
be a map
\[
\begin{array}{ccl}
V \times V & \to & \K \\
(v,w)   &\mapsto& \langle v,w \rangle
\end{array}
\]
which is $\K$-linear in the second argument and obeys the equation
\[      \langle v, w \rangle = \langle w, v \rangle^* \]
for all $v, w \in V$.  We say a hermitian form is {\bf nondegenerate}
if 
\[          \langle v, w \rangle = 0 
\textrm{ \; for \; all \; } v \in V \quad \iff \quad w = 0 .\]

Suppose $H$ is a $\K$-vector space with a nondegenerate hermitian
form.  Then we can define various concepts familiar from the theory
of Hilbert spaces.  For example, 
we say a sequence $e_i \in H$ is {\bf orthonormal} if $\langle e_i,
e_j \rangle = \delta_{ij}$. 
For any vector subspace $S \subseteq H$, we define
the {\bf orthogonal complement} $S^\perp$ by
\[    S^\perp = \{ w \in H \colon \; \langle v, w \rangle = 0 
\textrm{ \; for \; all \; } v \in S \} .\]
We say $S$ is {\bf closed} if $S^{\perp \perp} = S$.  Finally, we say $H$
is {\bf orthomodular} if $S + S^\perp = H$ for every closed subspace $S$. 
This is automatically true if $H$ is a real, complex, or quaternionic 
Hilbert space.

We are now ready to state Sol\`er's theorem: 
\begin{theorem} Let $\K$ be a division $\ast$-ring, and let
$H$ be a $\K$-vector space equipped with an orthomodular hermitian
form for which there exists an infinite orthonormal sequence.  Then
$\K$ is isomorphic to $\R, \C$ or $\H$, and $H$ is an
infinite-dimensional $\K$-Hilbert space.
\end{theorem}
\noindent 
Nothing in the assumptions mentions the continuum: the
hypotheses are purely algebraic.  It therefore seems quite magical
that $\K$ is forced to be the real numbers, complex
numbers or quaternions.  The orthomodularity condition is the key. 
In 1964, Piron had attempted to prove that any orthomodular
complex inner product space must be a Hilbert space~\cite{Piron2}.
His proof had a flaw which Ameniya and Araki fixed a few years 
later~\cite{AA}.  Their proof also handles the real and quaternionic
cases.  Eventually these ideas led Sol\`er to realize that orthomodularity
picks out real, complex, and quaternionic Hilbert spaces as special.

Holland's review article~\cite{Holland} describes the implications of
Sol\`er's results for other axiomatic approaches to quantum theory:
for example, approaches using convex cones and lattices.
This article is so well-written that there is no point in summarizing
it here.  Instead, we turn to some problems with real and
quaternionic quantum theory, which can be solved by treating them as
part of a larger structure---the `three-fold way'---that also includes
the complex theory.

\section{Problems}
\label{problems}

Is complex quantum theory `better' than the real or quaternionic
theories?  Of course it fits the experimental data better.  But people
have also pointed to various internal problems, or at least unfamiliar 
aspects, of the real and quaternionic theories, as a way to justify nature's
choice of the complex theory.  Here we describe two problems, both of which
can be solved using the `three-fold way'.  For more on the special virtues
of complex quantum theory see Hardy~\cite{Hardy} and Vicary~\cite{Vicary}.

One problem with real and quaternionic quantum theory is that the usual
correspondence between strongly continuous
one-parameter unitary groups and self-adjoint
operators breaks down, or at least becomes more subtle.  To see this
with a minimum of distractions, \textit{all Hilbert spaces in this
section will be assumed finite-dimensional}.  The same issues show up
in the infinite-dimensional case, but more technical detail is
required to tell the story correctly.

Suppose $\K$ is $\R,\C$ or $\H$.  Let $H$ and $H'$ be $\K$-Hilbert spaces.
We say an operator $T \maps H \to H'$ is {\bf bounded} if there 
exists a constant $K > 0$ such that
\[      \|Tv\| \le K \, \|v\|  \]
for all $v \in H$.          
We define the {\bf adjoint} of a bounded operator $T \maps H
\to H'$ in the usual way:
\[   \langle u, T^\dagger v \rangle = \langle T u, v \rangle \]
for all $u \in H$, $v \in H'$.   It is easy to check
that $T^\dagger$, defined this way, really is an operator from $H'$
back to $H$.   We define a $\K$-linear operator $U \maps H \to H'$ to 
be {\bf unitary} if $U U^\dagger = U^\dagger U = 1$.  Here we should
warn the reader that when $\K = \R$, the term `orthogonal' is often
used instead of `unitary'---and when $\K = \H$, people sometimes
use the term `symplectic'.  For us it will be more efficient to
use the same word in all three cases.

We define a {\bf one-parameter unitary group} to be a family of
unitary operators $U(t) \maps H \to H$, one for each $t \in \R$, such
that
\[            U(t + t') = U(t) U(t')  \]
for all $t, t' \in \R$.  If $H$ is finite-dimensional, every 
one-parameter unitary group can be written as 
\[           U(t) = \exp(t S)  \] 
for a unique bounded operator $S$.  

It is easy to check that the operator $S$ is {\bf skew-adjoint}: it 
satisfies $S^\dagger = -S$.  In quantum theory we usually want 
our observables to be {\bf self-adjoint} operators, obeying
$A^\dagger = A$.  And indeed, when $\K = \C$, we can write 
any skew-adjoint operator $S$ in terms of a
self-adjoint operator $A$ by setting
\[                  S(v) = A(v) i  .\]
This gives the usual correspondence between one-parameter unitary 
groups and self-adjoint operators.

Alas, when $\K = \R$ we have no number $i$, so we cannot express our 
skew-adjoint $S$ in terms of a self-adjoint $A$.
Nor can we do it when $\K = \H$.  Now the problem
is not a shortage of square roots of $-1$.  Instead, there are too
many---and more importantly, they do not commute!  We can try to set
$A(v) = S(v)i$, but this operator $A$ will rarely be linear.  The reason 
is that because $S$ commutes with multiplication by $j$, $A$ anticommutes
with multiplication by $j$, so $A$ is only linear in the trivial case $S = 0$.

Later we shall see a workaround for this problem.  A second problem,
special to the quaternionic case, concerns tensoring Hilbert spaces.
In ordinary complex quantum theory, when we have two systems, one
with Hilbert space $H$ and one with Hilbert space $H'$, the system
made by combining these two systems has Hilbert space $H \otimes H'$.
Here the tensor product of Hilbert spaces relies on a more primitive
concept: the tensor product of vector spaces.  The tensor product of
bimodules of a noncommutative algebra is another bimodule over that
algebra.  But a quaternionic vector space is not a bimodule: it is
only a left module of the quaternions.  So, the tensor product of
quaternionic Hilbert spaces is not naturally a quaternionic Hilbert
space.  

Unfortunately, this issue is not addressed head-on in Adler's book on
quaternionic quantum theory~\cite{Adler}.  When tensoring two quaternionic
Hilbert spaces, he essentially chooses a way to make one of them
into a bimodule, without being very explicit about this.  There is
no canonical way to make a quaternionic Hilbert space $H$ into a 
bimodule.  Indeed, given one way to do it, we can get 
many new ways as follows:
\[           (xv)_{\mathrm{new}} = \alpha(x) v , \qquad
             (vx)_{\mathrm{new}} = v x , \]
where $v \in H$, $x \in \H$ and $\alpha$ is an automorphism of the
quaternions.  Every automorphism is of the form
\[         \alpha(x) = g x g^{-1}  \]
for some unit quaternion $g$, so the automorphism group of the 
quaternions is $\SU(2)/\{\pm 1\} \cong \SO(3)$.  Thus we can 
`twist' a bimodule structure on $H$ by any element of $\SO(3)$, 
obtaining a new bimodule with the same underlying quaternionic
Hilbert space.
              
In Bartels' work on quaternionic functional analysis, he used
bimodules right from the start~\cite{Bartels}.  Later Ng~\cite{Ng} 
investigated these questions in more depth, and his paper provides 
a nice overview of many different approaches.  Here however we shall take a 
slightly different tack, noting that there is a way to make the tensor product
of quaternionic Hilbert spaces into a {\it real} Hilbert space.  This
fits nicely into the thesis we wish to advocate: that real, complex
and quaternionic quantum mechanics are not separate theories, but fit
together as parts of a unified whole.

\section{The Three-Fold Way}
\label{three-fold}

Some aspects of quantum theory become more visible when we introduce
{\it symmetry}.  This is especially true when it comes to the relation
between real, complex and quaternionic quantum theory.  So, instead of
bare Hilbert spaces, it is useful to consider Hilbert spaces equipped
with a representation of a group.  For simplicity suppose that $G$ is 
a Lie group, where we count a discrete group as a 0-dimensional Lie group.
Let $\Rep(G)$ be the category where:
\begin{itemize}
\item An object is a finite-dimensional complex Hilbert space $H$
equipped with a continuous unitary representation of $G$, 
say $\rho \maps G \to \U(H)$.  
\item A morphism is an operator $T \maps
H \to H'$ such that $T \rho(g) = \rho'(g) T$ for all $g \in G$.
\end{itemize}
To keep the notation simple, we often call a representation $\rho \maps
G \to \U(H)$ simply by the name of its underlying Hilbert space, $H$.
Another common convention, used in the introduction, is to call
it $\rho$.

Many of the operations that work for finite-dimensional Hilbert spaces 
also work for $\Rep(G)$, with the same formal properties: for example,
direct sums, tensor products and duals.  This lets us formulate
the {\bf three-fold way} as follows.   If an object $H \in \Rep(G)$
is \textbf{irreducible}---not a direct sum of other representations 
in a nontrivial way---there are three mutually exclusive choices:

\begin{itemize}
\item
The representation $H$ is not isomorphic to its dual.  In this
case we call it \textbf{complex}.
\item 
The representation $H$ is isomorphic to its dual and 
it is \textbf{real}: it arises from a representation of 
$G$ on a real Hilbert space $H_\R$ as follows:
\[             H = H_\R \otimes \C  \]
\item
The representation $H$ is isomorphic to its dual
and it is \textbf{quaternionic}: it comes from a representation 
of $G$ on a quaternionic Hilbert space $H_\H$ with
\begin{center}
$H =$ the underlying complex representation of $H_\H$
\end{center}
\end{itemize}

Why?  Suppose that $H \in \Rep(G)$ is irreducible.
Then there is a 1-dimensional space of
morphisms $f \maps H \to H$, by Schur's Lemma.  Since $H \iso H^*$,
there is also a 1-dimensional space of morphisms $T \maps H \to H^*$, and
thus a 1-dimensional space of morphisms
\[       g \maps H \otimes H \to \C \,  \]
We can also think of these as bilinear maps
\[       g \maps H \times H \to \C \]
that are invariant under the action of $G$ on $H$.  

But the representation $H \otimes H$ is the direct sum of two others:
the space $S^2 H$ of symmetric tensors, and the space $\Lambda^2 H$ of
antisymmetric tensors:
\[        H \otimes H \iso S^2 H \oplus \Lambda^2 H \,. \]
So, either there exists a nonzero $g$ that is {\bf symmetric}:
\[       g(v,w) = g(w,v)   \]
or a nonzero $g$ that is {\bf antisymmetric}:
\[      g(v,w) = -g(w,v) \, .  \]
One or the other, not both!---for if we had both, the space
of morphisms $g \maps H \otimes H \to \C$ would be at least
two-dimensional.

Either way, we can write
\[    g(v,w) = \langle J v, w \rangle \]
for some function $J \maps H \to H$.  Since $g$ and the inner product
are both invariant under the action of $G$, this function $J$ must commute
with the action of $G$.  But note that since $g$ is linear in
the first slot, while the inner product is not, $J$ must be 
\textbf{antilinear}, meaning
\[                J(vx + wy) = J(v)x^* + J(w)y^*  \]
for all $v,w \in V$ and $x,y \in \K$.   

The square of an antilinear operator is linear.   Thus, $J^2$ is 
linear and it commutes with the action of $G$.  By Schur's Lemma, it 
must be a scalar multiple of the identity:
\[  J^2 = c \]
for some $c \in \C$.  We wish to show that depending on whether
$g$ is symmetric or antisymmetric, we can rescale $J$ to achieve either
$J^2 = 1$ or $J^2 = -1$.  To see this, first note that
depending on whether $g$ is symmetric or antisymmetric, we have
\[   \pm g(v,w) = g(w,v) \]
and thus
\[   \pm \langle Jv, w\rangle = \langle Jw, v \rangle . \]
Choosing $v = Jw$, we thus have
\[   \pm \langle J^2 w,w\rangle = \langle Jw , Jw \rangle \ge 0. \]
It follows that $\pm J^2$ is a positive operator, so $c > 0$.  
Thus, dividing $J$ by
the positive square root of $c$, we obtain a new antilinear
operator --- let us again call it $J$ --- with $J^2 = \pm 1$.

Rescaled this way, $J$ is {\bf antiunitary}: it is an invertible
antilinear operator with 
\[    \langle Jv, Jw\rangle = \langle w, v \rangle . \]
Now consider the two cases:
\begin{itemize}
\item If $g$ is symmetric, $H$ is equipped with a {\bf real structure}: 
an antiunitary operator $J$ with
\[        J^2 = 1 .  \]
Furthermore $J$ commutes with the action of $G$.  
It follows that the real Hilbert space
\[       H_\R = \{  x \in H \colon \; Jx = x  \} \]
is a representation of $G$ whose tensor product with $\C$ is
$H$.
\item
If $g$ is antisymmetric, $H$ is equipped with a {\bf quaternionic
structure}: an antiunitary operator $J$ with 
\[       J^2 = -1 . \]
Furthermore $J$ commutes with the action of $G$.
Since the operators $I = i$, $J$, and $K = IJ$ obey the 
usual quaternion relations, $H$ can be made into a quaternionic
Hilbert space $H_\H$.  There is a representation of $G$ on $H_\H$ 
whose underlying complex representation is $H$.
\end{itemize}
\noindent
Note that in both cases, $g$ is {\bf nondegenerate}, meaning
\[        \forall v \in V \;\;  g(v,w) = 0 \quad \implies \quad w = 0 .\]
The reason is that $g(v,w) = \langle J v, w \rangle$, and the inner
product is nondegenerate, while $J$ is one-to-one.  

We can describe real and quaternionic representations using either
$g$ or $J$.  In the following statements, we do not need the representation
to be irreducible:

\begin{itemize}
\item
Given a complex Hilbert space $H$, a nondegenerate
symmetric bilinear form $g \maps H \times H \to \C$ is called an 
{\bf orthogonal structure} on $H$.  So, a representation 
$\rho \maps G \to \U(H)$ is real iff it preserves some orthogonal
structure $g$ on $H$.  Alternatively, $\rho$ is real iff there is
a real structure $J \maps H \to H$ that commutes with the 
action of $G$.  
\item
Similarly, given a complex Hilbert space $H$, a nondegenerate
skew-symmetric bilinear form $g \maps H \times H \to \C$ is called a
{\bf symplectic structure} on $H$.  So, a representation
$\rho \maps G \to \U(H)$ is quaternionic iff it preserves some 
symplectic structure $g$ on $H$.  Alternatively, 
$\rho$ is quaternionic iff there is a quaternionic structure 
$J \maps H \to H$ that commutes with the action of $G$.  
\end{itemize}

We can summarize these patterns as follows:

\begin{center}
THE THREEFOLD WAY
\end{center}
\vskip 1em
\begin{center}
\begin{tabular}{ccccccc}                    
complex  &$\quad$ & $H \ncong H^*$  &$\quad$& unitary   
\\    \\
real   && $H \cong H^*$  && orthogonal  \\
       && $J^2 = 1$  
\\  \\
quaternionic   && $H \cong H^*$ && symplectic \\
               && $J^2 = -1$  
\end{tabular} 
\end{center}
\vskip 1em
\noindent

\noindent
For more on how these patterns pervade mathematics, see Dyson's paper
on the three-fold way~\cite{Dyson}, and Arnold's paper on mathematical
`trinities'~\cite{Arnold}.  In the next section we focus on a few
applications to physics.

\section{Applications}
\label{physics}

Let us consider an example: $G = \SU(2)$.  In physics this group
is important because its representations describe the ways
a particle can transform under rotations.  There is one one irreducible 
representation for each `spin' $j = 0, \frac{1}{2}, 1, \dots$.  
When $j$ is an integer, this representation describes the angular
momentum states of a boson; when $j$ is a half-integer (meaning an
integer plus $\frac{1}{2}$) it describes the angular momentum states
of a fermion.  

Mathematically, $\SU(2)$ is special in many ways.  For one thing,
all its representations are self-dual.   When $j$ is an integer, 
the spin-$j$ representation is real; when $j$ is an integer plus 
$\frac{1}{2}$, the spin-$j$ representation is quaternionic.   
We shall see why this is true shortly, but we cannot resist making 
a little table that summarizes the pattern:

\begin{center}
IRREDUCIBLE REPRESENTATIONS OF $\SU(2)$
\end{center}
\vskip 0.5em
\begin{center}
\begin{tabular}{lcclc}
$j \in \Z$ & bosonic   & real         & $J^2 = 1$  & orthogonal  \\
$j \in \Z + \frac{1}{2}$ & fermionic & quaternionic & $J^2 = -1$ & symplectic 
\end{tabular} 
\end{center}
\vskip 1em

\noindent
By the results described in the last section, we conclude 
that the integer-spin
or `bosonic' irreducible representations of $\SU(2)$ are
equipped with an invariant antiunitary operator $J$ with $J^2 = 1$,
and an invariant orthogonal structure.  The half-integer-spin 
or `fermionic' ones are equipped with an invariant antiunitary 
operator $J$ with $J^2 = -1$, and an invariant symplectic structure.

Why does it work this way?  First, consider the spin-$\frac{1}{2}$ 
representation.   We can identify the group of unit quaternions 
(that is, quaternions with norm 1) with $\SU(2)$ as follows:
\[     a + bi + cj + dk \; \mapsto \;  a \sigma_0 - ib \sigma_1 -
ic \sigma_2 - id  \sigma_3  \]
where the $\sigma$'s are the Pauli matrices.  This lets us think of 
$\C^2$ as the underlying complex vector space of $\H$
in such a way that the obvious representation
of $\SU(2)$ on $\C^2$ gets identified with the action of unit quaternions
on $\H$ by left multiplication.  Since left multiplication commutes with
right multiplication, this action is quaternion-linear.  
Thus, the spin-$\frac{1}{2}$ representation of $\SU(2)$ is quaternionic.

Next, note that can build all the irreducible unitary representations of 
$\SU(2)$ by forming symmetrized tensor powers of the spin-$\frac{1}{2}$
representation.  The $n$th symmetrized tensor power, $S^n(\C^2)$, is 
the spin-$j$ representation with $j = n/2$.     
At this point a well-known general result will help:

\begin{theorem} 
\label{tensor} 
Suppose $G$ is a Lie group and $H, H' \in \Rep(G)$.  Then:
\begin{itemize}
\item
$H$ and $H'$ are real $\implies$ $H \otimes H'$ is real.
\item
$H$ is real and $H'$ is quaternionic $\implies$ $H \otimes H'$ is 
quaternionic.
\item
$H$ is quaternionic and $H'$ is real $\implies$ $H \otimes H'$ is
quaternionic.
\item
$H$ and $H'$ are quaternionic $\implies$ $H \otimes H'$ is real.
\end{itemize}
\end{theorem}
\begin{proof}  As we saw in the last section,
$H$ is real (resp.\ quaternionic) if and only if
it can be equipped with an invariant antiunitary operator $J \maps H \to
H$ with $J^2 = 1$ (resp.\ $J^2 = -1$).  So, pick such an antiunitary
$J$ for $H$ and also one $J'$ for $H'$.  Then $J \tensor J'$ is
an invariant antiunitary operator on $H \tensor H'$, and 
\[        (J \tensor J')^2 = J^2 \tensor {J'}^2 .\]
This makes the result obvious.  
\end{proof}
\noindent In short, the `multiplication table' for tensoring real
and quaternionic representations is just like the multiplication 
table for the numbers $1$ and $-1$---and also, not coincidentally, 
just like the usual rule for combining bosons and fermions.  

Next, note that any subrepresentation of a real representation is
real, and any subrepresentation of a quaternionic representation is
quaternionic.  It follows that since the spin-$\frac{1}{2}$ 
representation of $\SU(2)$ is quaternionic, its $n$th tensor 
power is real or quaternionic depending on whether $n$ is even or
odd, and similarly for the subrepresentation $S^n(\C^2)$.  
This is the spin-$j$ representation for $j = n/2$.  So, the
spin-$j$ representation is real or quaternionic depending on
whether $j$ is an integer or half-integer.

But what is the physical meaning of the antiunitary operator
$J$ on the spin-$j$ representation?   Let us call this representation
$\rho \maps \SU(2) \to \U(H)$, where $H$ is the Hilbert space $S^n(\C^2)$.
Choose any element $X \in \su(2)$ and let $S = d\rho(X)$.  Then
\[            J \, \exp(tS) =  \exp(tS) \, J \]
for all $t \in \R$, and thus differentiating, we obtain
\[            J S = S J .\]
But the operator $S$ is skew-adjoint.  In quantum mechanics,
we write $S = iA$ where $A$ is the self-adjoint operator corresponding
to angular momentum along some axis.  Since $J$ anticommutes 
with $i$, we have
\[             J A = -A J  \]
so for every $v \in H$ we have
\[   \langle  J v, A J v \rangle = - \langle J v, J A v \rangle =
- \langle A v, v \rangle = - \langle v, A v \rangle .\]
Thus, the antiunitary operator $J$ reverses angular momentum.  Since
the time-reversed version of a spinning particle is a particle
spinning the opposite way, physicists call $J$ \textbf{time reversal}.  

It seems natural that time reversal should obey $J^2 = 1$, as it does
in the bosonic case; it may seem strange to have $J^2 = -1$, as 
we do for fermions.  Should not applying time reversal twice get us
back where we started?  It actually does, in a crucial sense: the expectation
values of all observables are unchanged when we multiply a unit vector
$v \in H$ by $-1$.  Still, this minus sign reminds one of the equally 
curious minus sign that a fermion picks up when one rotates it a full turn.  
Are these signs related?

The answer appears to be yes, though the following argument is
more murky than we would like.  It has its origins in Feynman's 1986
Dirac Memorial Lecture on antiparticles~\cite{Feynman}, together with
the well-known `belt trick' for demonstrating that $\SO(3)$ fails to
be simply connected (which is the reason its double cover $\SU(2)$ is
needed to describe rotations in quantum theory.)  The argument uses 
string diagrams, which are a generalization of Feynman diagrams.  
Instead of explaining string diagrams here, we refer the reader to our 
prehistory of $n$-categorical physics~\cite{BL}, and especially
the section on Freyd and Yetter's 1986 paper on tangles.

The basic ideas apply to any Lie group $G$.  If an irreducible
object $H \in \Rep(G)$ is isomorphic to its dual, it comes with 
an invariant nondegenerate pairing $g \maps H \otimes H \to \C$. 
We can draw this as a `cup':
\[
 \xy 0 ;/r.15pc/:
    (-5,8)*{}="x1";
    (5,8)*{}="x2";
    (-5,5)*{}="m1";
    (5,5)*{}="m2";
    (-5,-5)*{}="k1";
    (5,-5)*{}="k2";
    (-5,-8)*{}="y1";
    (5,-8)*{}="y2";
 "m1";"k1" **\dir{-};
 "m2";"k2" **\dir{-};
 "k1";"y1" **\dir{-};
 "k2";"y2" **\dir{-};
 "y1";"y2" **\crv{(-5,-14) & (5,-14)};
\endxy
\]
If this were a Feynman diagram, it would depict two particles of type 
$H$ coming in on top and nothing going out at the bottom: since 
$H$ is isomorphic to its dual, particles of this type act like their 
own antiparticles, so they can annihilate each other.    

Since $g$ defines an isomorphism $H \iso H^*$, we get a
morphism back from $\C$ to $H \otimes H$, which we draw as 
a `cap':
\[
 \xy 0 ;/r.15pc/:
    (-5,-8)*{}="x1";
    (5,-8)*{}="x2";
    (-5,-5)*{}="m1";
    (5,-5)*{}="m2";
    (-5,5)*{}="k1";
    (5,5)*{}="k2";
    (-5,8)*{}="y1";
    (5,8)*{}="y2";
 "m1";"k1" **\dir{-};
 "m2";"k2" **\dir{-};
 "k1";"y1" **\dir{-};
 "k2";"y2" **\dir{-};
 "y1";"y2" **\crv{(-5,14) & (5,14)};
\endxy
\]
uniquely determined by requiring that the cap and cup 
obey the \textit{zig-zag identities:}
\[
\xy 0 ;/r.15pc/:
    (-10,-12)*{}="1E";
    (-10,0)*{}="1";
    (0,0)*{}="2";
    (10,0)*{}="3";
    (10,12)*{}="3B";
 "2";"1" **\crv{(0,10)& (-10,10)}
     ?(.03)*\dir{-}  ?(1)*\dir{-};
 "3";"2" **\crv{(10,-10)& (0,-10)}
     ?(.03)*\dir{-}  ;
 "1";"1E" **\dir{-};
 "3B";"3" **\dir{-};
     (-5,8.5)*{\scriptstyle };
     (5,-9)*{\scriptstyle };
\endxy
\quad = \quad
\xy ;/r.15pc/:
(-6,12)*{}; (0,12)*{}; (0,-12)*{}; **\dir{-}
?(.47)*\dir{-}; (6,-12)*{}; (4,0)*{\scriptstyle }
\endxy
\]
\[
  \xy 0 ;/r.15pc/:
    (-10,12)*{}="1E";
    (-10,0)*{}="1";
    (0,0)*{}="2";
    (10,0)*{}="3";
    (10,-12)*{}="3B";
 "2";"1" **\crv{(0,-10)& (-10,-10)}
     ?(.03)*\dir{-}  ?(1)*\dir{-};
 "3";"2" **\crv{(10,10)& (0,10)}
     ?(.03)*\dir{-}  ;
 "1";"1E" **\dir{-};
 "3B";"3" **\dir{-};
      (5,10)*{\scriptstyle  };
     (-5,-10.5)*{\scriptstyle };
\endxy
\quad = \quad
\xy ;/r.15pc/:
(-6,12)*{}; (0,12)*{}; (0,-12)*{}; **\dir{-}
?(.47)*\dir{-}; (6,-12)*{}; (4,0)*{\scriptstyle }
\endxy
\]
Indeed, for any $H \in \Rep(G)$, we have a `cup' $H^* \tensor H \to \C$
and `cap' $\C \to H \tensor H^*$ obeying the zig-zag identities: in the
language of category theory, we say that $\Rep(G)$ is compact 
closed~\cite{BL}.  When $H$ is isomorphic to its dual, we can write these using
just $H$. 

As we have seen, there are are two choices.  If $H$ is real, $g$ is 
an orthogonal structure, and we can write the equation $g(w,v) = g(v,w)$ 
using string diagrams as follows:
\[
 \xy  0 ;/r.15pc/:
    (5,8)*{}="x1";
    (-5,8)*{}="x2";
    (5,5)*{}="m1";
    (-5,5)*{}="m2";
    (5,-5)*{}="k1";
    (-5,-5)*{}="k2";
    (5,-8)*{}="y1";
    (-5,-8)*{}="y2";
 \vtwist~{"m1"}{"m2"}{"k1"}{"k2"};
 "x1";"m1" **\dir{-}; 
 "x2";"m2" **\dir{-}; 
 "k1";"y1" **\dir{-}; 
 "k2";"y2" **\dir{-}; 
    "y1";"y2" **\crv{(5,-14) & (-5,-14)};
\endxy
\quad = \quad
 \xy 0 ;/r.15pc/:
    (-5,8)*{}="x1";
    (5,8)*{}="x2";
    (-5,5)*{}="m1";
    (5,5)*{}="m2";
    (-5,-5)*{}="k1";
    (5,-5)*{}="k2";
    (-5,-8)*{}="y1";
    (5,-8)*{}="y2";
 "m1";"k1" **\dir{-}; 
 "m2";"k2" **\dir{-}; 
 "x1";"m1" **\dir{-}; 
 "x2";"m2" **\dir{-}; 
 "k1";"y1" **\dir{-}; 
 "k2";"y2" **\dir{-}; 
    "y1";"y2" **\crv{(-5,-14) & (5,-14)};
\endxy
\]
If $H$ is quaternionic, $g$ is a symplectic structure, and we can
write the equation $g(w,v) = -g(v,w)$ as follows:
\[
 \xy  0 ;/r.15pc/:
    (5,8)*{}="x1";
    (-5,8)*{}="x2";
    (5,5)*{}="m1";
    (-5,5)*{}="m2";
    (5,-5)*{}="k1";
    (-5,-5)*{}="k2";
    (5,-8)*{}="y1";
    (-5,-8)*{}="y2";
 \vtwist~{"m1"}{"m2"}{"k1"}{"k2"};
 "x1";"m1" **\dir{-}; 
 "x2";"m2" **\dir{-}; 
 "k1";"y1" **\dir{-}; 
 "k2";"y2" **\dir{-}; 
    "y1";"y2" **\crv{(5,-14) & (-5,-14)};
\endxy
\quad
= 
\quad - \; \;
 \xy 0 ;/r.15pc/:
    (-5,8)*{}="x1";
    (5,8)*{}="x2";
    (-5,5)*{}="m1";
    (5,5)*{}="m2";
    (-5,-5)*{}="k1";
    (5,-5)*{}="k2";
    (-5,-8)*{}="y1";
    (5,-8)*{}="y2";
 "m1";"k1" **\dir{-}; 
 "m2";"k2" **\dir{-}; 
 "x1";"m1" **\dir{-}; 
 "x2";"m2" **\dir{-}; 
 "k1";"y1" **\dir{-}; 
 "k2";"y2" **\dir{-}; 
    "y1";"y2" **\crv{(-5,-14) & (5,-14)};
\endxy
\]

The power of string diagrams is that we can apply the same diagrammatic
manipulation to both sides of an equation and get a valid new equation.  So
let us take the equations above, turn both sides $90^\circ$ clockwise, 
and stretch out the string a little.  When $H$ is a real representation, 
we obtain:
\[ \xy 0 ;/r.20pc/:
  (0,15)*{}="T";
  (0,-15)*{}="B";
  (0,7.5)*{}="T'";
  (0,-7.5)*{}="B'";
  "T";"T'" **\dir{-}?(.5)*\dir{-};
    "B";"B'" **\dir{-};
    (-4.5,0)*{}="MB";
    (-10.5,0)*{}="LB";
    "T'";"LB" **\crv{(-1.5,-6) & (-10.5,-6)}; \POS?(.25)*+{\hole}="2z";
    "LB"; "2z" **\crv{(-12,9) & (-3,9)};
    "2z";"B'" **\crv{(0,-4.5)};
    (2,13)*{};
    \endxy
    \qquad = \qquad 
  \xy 0 ;/r.20pc/:
  (0,15)*{}="T";
  (0,-15)*{}="B";
  (0,7.5)*{}="T'";
  (0,-7.5)*{}="B'";
  "T";"B" **\dir{-}?(.5)*\dir{-};
    (2,13)*{};
    \endxy
 \]
When $H$ is quaternionic we obtain:
\[ \xy 0 ;/r.20pc/:
  (0,15)*{}="T";
  (0,-15)*{}="B";
  (0,7.5)*{}="T'";
  (0,-7.5)*{}="B'";
  "T";"T'" **\dir{-}?(.5)*\dir{-};
    "B";"B'" **\dir{-};
    (-4.5,0)*{}="MB";
    (-10.5,0)*{}="LB";
    "T'";"LB" **\crv{(-1.5,-6) & (-10.5,-6)}; \POS?(.25)*+{\hole}="2z";
    "LB"; "2z" **\crv{(-12,9) & (-3,9)};
    "2z";"B'" **\crv{(0,-4.5)};
    (2,13)*{};
    \endxy
    \qquad = \quad \; - \; \;
  \xy 0 ;/r.20pc/:
  (0,15)*{}="T";
  (0,-15)*{}="B";
  (0,7.5)*{}="T'";
  (0,-7.5)*{}="B'";
  "T";"B" **\dir{-}?(.5)*\dir{-};
    (2,13)*{};
    \endxy
 \]
Both sides of these equations describe morphisms from $H$ to itself.
The vertical straight line at right corresponds to the identity 
morphism $1 \maps H \to H$.  The more complicated diagram at left
turns out to correspond to the morphism $J^2$.  So, it follows that
$J^2 = 1$ when the representation $H$ is real, while $J^2 = -1$ 
when $H$ is quaternionic.  

While it may seem puzzling to those who have not been initiated 
into the mysteries of string diagrams, everything so far can be
made perfectly rigorous~\cite{HDA2}.  Now comes the murky part.  
The string diagram at left: 
\[ \xy 0 ;/r.20pc/:
  (0,15)*{}="T";
  (0,-15)*{}="B";
  (0,7.5)*{}="T'";
  (0,-7.5)*{}="B'";
  "T";"T'" **\dir{-}?(.5)*\dir{-};
    "B";"B'" **\dir{-};
    (-4.5,0)*{}="MB";
    (-10.5,0)*{}="LB";
    "T'";"LB" **\crv{(-1.5,-6) & (-10.5,-6)}; \POS?(.25)*+{\hole}="2z";
    "LB"; "2z" **\crv{(-12,9) & (-3,9)};
    "2z";"B'" **\crv{(0,-4.5)};
    (2,13)*{};
\endxy
 \]
looks like a particle turning back in time, going backwards 
for a while, and then turning yet again and continuing forwards
in time.  In other words, it {\it looks like a picture of the 
square of time reversal!}   So, the fact that the corresponding 
morphism $J^2$ is indeed the square of time reversal when 
$G = \SU(2)$ somehow seems right.  

Moreover, in the theory of string diagrams it is often useful to 
draw the strings as `ribbons'.  This allows us to take the diagram at
left and pull it tight, as follows:
\[
 \xy
 (0,0)*{\includegraphics[angle=180,scale=0.5]{ribbon1.eps}};
 \endxy
\qquad = \qquad
 \xy
 (0,0)*{\includegraphics[angle=180,scale=0.5]{ribbon2.eps}};
 \endxy
\]
At left we have a particle (or actually a small piece of string) 
turning around in {\it time}, while at right we have a particle 
making a full turn in {\it space} as time passes.  So, in the case
$G = \SU(2)$ it seems to make sense that $J^2 = 1$ when 
the spin $j$ is an integer: after all, the state vector of an
integer-spin particle is unchanged when we rotate that particle a full 
turn.  Similarly, it seems to make sense that $J^2 = -1$ when $j$ is 
a half-integer: the state vector of a half-integer-spin particle gets
multiplied by $-1$ when we rotate it by a full turn.   

So, we seem to be on the brink of having a `picture proof' that 
the square of time reversal must match the result of turning
a particle 360 degrees.  Unfortunately this argument is not yet 
rigorous, since we have not explained how the topology of ribbon diagrams
(well-known in category theory) is connected to the geometry of 
rotations and time reversal.  Perhaps understanding the argument
better would lead us to new insights.  

While our discussion here focused on the group $\SU(2)$, real and
quaternionic representations of other groups are also important
in physics.  For example, gauge bosons live in the adjoint representation
of a compact Lie group $G$ on the complexification of its own Lie algebra;
since the Lie algebra is real, this is always a real representation of
$G$.  This is related to the fact that gauge bosons are their own 
antiparticles.  In the Standard Model, fermions are not their own 
antiparticles, but in some theories they can be.  Among other things,
this involves the question of whether the relevant spinor representations 
of the groups $\Spin(p,q)$ are complex, real (`Majorana spinors') or 
quaternionic (`pseudo-Majorana spinors').  The options are 
well-understood, and follow a nice pattern depending on the dimension 
and signature of spacetime modulo 8~\cite{BT}.  We should emphasize
that the spin groups $\Spin(p,q)$ are not compact unless $p = 0$ or 
$q = 0$, so their finite-dimensional complex representations are hardly 
ever unitary, and most of the results mentioned in this paper 
do not apply.  Nonetheless, we may ask if such a representation is
the complexification of a real one, or the underlying complex 
representation of a quaternionic one---and the answers have
implications for physics.

$\SU(2)$ is not the only compact Lie group with the property that
all its irreducible continuous unitary representations on complex 
Hilbert spaces are real or quaternionic.  For a group to have this 
property, it is necessary and sufficient that every element 
be conjugate to its inverse.  All compact simple Lie groups have
this property except those of type $A_n$ for $n > 1$, $D_n$ with 
$n$ odd, and $E_6$ (see Bourbaki~\cite{Bourbaki2}).  For the 
symmetric groups $S_n$, the orthogonal groups $\OO(n)$, and the 
special orthogonal groups $\SO(n)$ for $n \ge 3$, all representations 
are in fact real.  So, there is a rich supply of real and quaternionic group 
representations, which leave their indelible mark on physics even 
when we are doing complex quantum theory.

\section{Categories}
\label{categories}

As discussed so far, the three-fold way describes how certain 
complex representations of groups can be seen as arising from
real or quaternionic representations.  This gives a sense in which
ordinary complex quantum theory subsumes the real and quaternionic
theories.  But underlying this idea is something simpler yet deeper, 
which arises at the level of Hilbert spaces, even before group 
representations enter the game.  Suppose $\K$ is either $\R$, $\C$ 
or $\H$, and let $\Hilb_\K$ be the category where:
\begin{itemize}
\item
an object is a $\K$-Hilbert space;
\item
a morphism is a bounded $\K$-linear operator.
\end{itemize}  
As we shall see, any one of the categories $\Hilb_\R$, $\Hilb_\C$ 
and $\Hilb_\H$ has a faithful functor to any other.  This means that a 
Hilbert space over any one of the three normed division algebras can be seen 
as Hilbert space over any other, equipped with some extra structure!  

At the bottom of this fact lies the chain of inclusions
\[\R \subset \C \subset \H  .\]
Thanks to these, any complex vector space has an underlying real vector 
space, and any quaternionic vector space has underlying complex 
vector space.   The same is true for Hilbert spaces.
To make the underlying real vector space of a complex Hilbert space
into a real Hilbert space, we take the real part of the original complex
inner product.  This is well-known; a bit less familiar is how the 
underlying complex vector space of a quaternionic Hilbert space becomes 
a complex Hilbert space.  Here we need to take the {\bf complex part}
of the original quaternionic inner product, defined by
\[          \mathrm{Co}(a + bi + cj + dk) = a + bi .\]
One can check that these constructions give functors
\[          \Hilb_\H \to \Hilb_\C \to \Hilb_\R  .\]

A bit more formally, we have a commutative triangle of homomorphisms:
\[          
\xymatrix{
\R \ar[dr]_{\alpha} \ar[rr]^{\gamma} && \H  \\
& \C \ar[ur]_{\beta}
}
\]
There is only one choice of the homomorphisms $\alpha$ and $\gamma$.
There are many choices of $\beta$, since we can map $i \in \C$ to any
square root of $-1$ in the quaternions.  However, all the various
choices of $\beta$ are the same up to an automorphism of the quaternions.
That is, given two homomorphisms $\beta,\beta' \maps \C \to \H$, we
can always find an automorphism $\theta \maps \H \to \H$ such that $\beta'
= \theta \circ \beta$.  So, nothing important depends on the choice of
$\beta$.  Let us make a choice---say the standard one, with $\beta(i) =
i$---and use that.  This commutative triangle gives a commutative 
triangle of functors:
\[  
\xymatrix{
\Hilb_{\R} && \Hilb_{\H} \ar[dl]^{\beta^*} \ar[ll]_{\gamma^*} \\
& \Hilb_{\C} \ar[ul]^{\alpha^*}
}
\]

Now, recall that a functor $F \maps C \to D$ is {\bf faithful} if
given two morphisms $f,f' \maps c \to c'$ in $C$, $F(f) = F(f')$
implies that $f = f'$.  When $F \maps C \to D$ is faithful, we 
can think of objects of $C$ as objects of $D$ equipped with extra structure.

It is easy to see that the functors $\alpha^*, \beta^*$ and $\gamma^*$
are all faithful.  This lets us describe Hilbert spaces for a larger
normed division algebra as Hilbert spaces {\it
equipped with extra structure} for a smaller normed division algebra.  
None of this particularly new or
difficult: the key ideas are all in Adams' book on Lie groups~\cite{Adams}.  
However, it sheds light on how real, complex and
quaternionic quantum theory are related.

First we consider the extra structure possessed by the 
underlying real Hilbert space of a complex Hilbert space:
\begin{theorem}
The functor $\alpha^* \maps \Hilb_\C \to \Hilb_\R$ is faithful,
and $\Hilb_\C$ is equivalent to the category where:
\begin{itemize}
\item an object is a real Hilbert space $H$ equipped with a
unitary operator $J \maps H \to H$ with $J^2 = -1$;
\item a morphism $T \maps H \to H'$ is a bounded real-linear operator
preserving this exta structure: $T J = J' T$.
\end{itemize}
\end{theorem}
\noindent This extra structure $J$ is often called a {\bf complex structure}.
It comes from our ability to multiply by $i$ in a complex Hilbert space.

Next we consider the extra structure possessed by the underlying
complex Hilbert space of a quaternionic Hilbert space.  For this
we need to generalize the concept of an antiunitary operator.  
First, given $\K$-vector spaces $V$ and $V'$, we define an 
{\bf antilinear operator} $T \maps V \to V'$ to be a map with 
\[                T(vx + wy) = T(v)x^* + T(w)y^*  \]
for all $v,w \in V$ and $x,y \in \K$.   Then, given $\K$-Hilbert spaces
$H$ and $H'$, we define an {\bf antiunitary operator} $T \maps H \to H'$ 
to be an invertible antilinear operator with 
\[            \langle T v, T w \rangle = \langle w, v \rangle \]
for all $v,w \in H$.  
\begin{theorem}
\label{quaternionic_as_complex}
The functor $\beta^* \maps \Hilb_\H \to \Hilb_\C$ is faithful,
and $\Hilb_\H$ is equivalent to the category where:
\begin{itemize}
\item an object is a complex Hilbert space $H$ equipped with an
antiunitary operator $J \maps H \to H$ with $J^2 = -1$;
\item a morphism $T \maps H \to H'$ is a bounded complex-linear operator
preserving this extra structure: $T J = J' T$.
\end{itemize}
\end{theorem}
\noindent This extra structure $J$ is often called a {\bf quaternionic 
structure}.  We have seen it already in our study of the three-fold way.
It comes from our ability to multiply by $j$ in a quaternionic Hilbert
space.

Finally, we consider the extra structure possessed by the underlying
real Hilbert space of a quaternionic Hilbert space.  This can be
understood by composing the previous two theorems:
\begin{theorem}
The functor $\gamma^* \maps \Hilb_\H \to \Hilb_\R$ is faithful,
and $\Hilb_\H$ is equivalent to the category where:
\begin{itemize}
\item an object is a real Hilbert space $H$ equipped with two 
unitary operators $J,K \maps H \to H$ with $J^2 = K^2 = -1$ and
$JK = -KJ$;
\item a morphism $T \maps H \to H'$ is a bounded real-linear operator
preserving this extra structure: $T J = J' T$ and $T K = K' T$.
\end{itemize}
\end{theorem}
\noindent This extra structure could also be called a 
{\bf quaternionic structure}, as long as we remember that a 
quaternionic structure on a real Hilbert space is different 
than one on a complex Hilbert space.  It comes from our ability
to multiply by $j$ and $k$ in a quaternionic Hilbert space.  
Of course if we define $I = JK$, then $I, J,$ and $K$ obey the
usual quaternion relations.

The functors discussed so far all have adjoints, which are in
fact both left and right adjoints:
\[          
\xymatrix{
\Hilb_\R \ar[dr]_{\alpha_*} \ar[rr]^{\gamma_*} && \Hilb_\H  \\
& \Hilb_\C \ar[ur]_{\beta_*}
}
\]
These adjoints can easily be defined using the theory of 
bimodules~\cite{AF}.  As vector spaces, we have:
\[
\begin{array}{ccl}
\alpha_*(V) &=& V \otimes {{}_\R\C_{\C}}  \\
\beta_*(V)   &=& V \otimes {{}_\C\H_{\H}}  \\
\gamma_*(V)  &=& V \otimes {{}_\R\H_{\H}}  \\
\end{array} 
\]
In the first line here, $V$ is a real vector space, or in other words,
a right $\R$-module, while ${{}_\R\C_{\C}}$ denotes $\C$ regarded as a
$\R$-$\C$-bimodule.  Tensoring these, we obtain a right
$\C$-module, which is the desired complex vector space.  The other
lines work analogously.  It is then easy to make all these vector
spaces into Hilbert spaces.  We cannot resist mentioning that our
previous functors can be described in a similar way, just by turning
the bimodules around:
\[
\begin{array}{ccl}
\alpha^*(V)  &=& V \otimes {{}_\C\C_{\R}}  \\
\beta^*(V)   &=& V \otimes {{}_\H\H_{\C}}  \\
\gamma^*(V)  &=& V \otimes {{}_\H\H_{\R}} 
\end{array} 
\]
But instead of digressing into this subject, which is called Morita
theory~\cite{AF}, we wish merely to
note that the functors $\alpha_*, \beta_*$ and $\gamma_*$ are also
faithful.  This lets us describe Hilbert spaces for a smaller
normed division algebra in terms of Hilbert spaces for a larger one.

We begin with the functor $\alpha_*$, which is called
{\bf complexification}:
\begin{theorem}
\label{real_as_complex}
The functor $\alpha_* \maps \Hilb_\R \to \Hilb_\C$ is faithful,
and $\Hilb_\R$ is equivalent to the category where:
\begin{itemize}
\item an object is a complex Hilbert space $H$ equipped with a
antiunitary operator $J \maps H \to H$ with $J^2 = 1$;
\item a morphism $T \maps H \to H'$ is a bounded complex-linear operator
preserving this exta structure: $T J = J' T$.
\end{itemize}
\end{theorem}
\noindent
The extra structure $J$ here is often called a {\bf real structure}.
We have seen it already in our study of the three-fold way.
It is really just a version of complex conjugation.  In other
words, suppose that $H$ is a real Hilbert space.  Then any element 
of its complexification can be written uniquely as
$u + v i$ with $u,v \in H$, and then
\[         J(u + v i) = u - v i .\]

Next we consider the functor $\beta_*$ from complex to quaternionic Hilbert
spaces:
\begin{theorem}
The functor $\beta_* \maps \Hilb_\C \to \Hilb_\H$ is faithful,
and $\Hilb_\C$ is equivalent to the category where:
\begin{itemize}
\item an object is a quaternionic Hilbert space $H$ equipped with a
unitary operator $J$ with $J^2 = -1$;
\item a morphism $T \maps H \to H'$ is a bounded quaternion-linear operator
preserving this extra structure: $T J = J' T$.
\end{itemize}
\end{theorem}
\noindent The operator $J$ here comes from our ability to multiply
any vector in $\beta_*(H) = H \otimes {{}_\C\H_{\H}}$ on the left by  
the complex number $i$.  This result is less well-known than the 
previous ones, so we sketch a proof:

\begin{proof} Suppose $H$ is a quaternionic Hilbert space equipped
with a unitary operator $J$ with $J^2 = -1$.  Then $J$ makes
$H$ into a right module over the complex numbers, and this action of 
$\C$ commutes with the action of $\H$, so $H$ becomes a right module
over the tensor product of $\C$ and $\H$, considered as algebras
over $\R$.  But this is isomorphic to the algebra of 
$2 \times 2$ complex matrices~\cite{Adams}.   The matrix
\[    
p = \left(\begin{array}{cc}
            1 & 0 \\
            0 & 0
\end{array}\right) \]
projects $H$ down to a complex Hilbert space $H_\C$ whose complex
dimension matches the quaternionic dimension of $H$.  By applying
arbitrary $2 \times 2$ complex matrices to this subspace we obtain
all of $H$, so $\beta_* H_\C$ is naturally isomorphic to $H$.   
\end{proof}

Composing $\alpha_*$ and $\beta_*$, we obtain the functor from real to 
quaternionic Hilbert spaces.  This is sometimes called {\bf 
quaternification}:
\begin{theorem}
The functor $\gamma_* \maps \Hilb_\R \to \Hilb_\H$ is faithful,
and $\Hilb_\R$ is equivalent to the category where:
\begin{itemize}
\item an object is a quaternionic Hilbert space $H$ equipped with two
unitary operators $J, K$ with $J^2 = K^2 = -1$ and $JK = -KJ$;
\item a morphism $T \maps H \to H'$ is a bounded quaternion-linear operator
preserving this extra structure: $T J = J' T$ and $T K = K' T$.
\end{itemize}
\end{theorem}
\noindent
The operators $J$ and $K$ here arise from our ability to multiply
any vector in $\gamma_*(H) = H \otimes {{}_\R\H_{\H}}$ on the left 
by the quaternions $j$ and $k$.  Again, we sketch a proof of this result:

\begin{proof}
The operators $J, K$ and $I = JK$ make $H$ into a left
$\H$-module.  Since this action of $\H$ commutes with the existing
right $\H$-module structure, $H$ becomes a module over the tensor
product of $\H$ and $\H^{\rm op} \cong \H$, considered as algebras
over $\R$.  But this tensor product is isomorphic to the algebra of $4
\times 4$ real matrices~\cite{Adams}.  The matrix
\[    
p = \left(\begin{array}{cccc}
            1 & 0 & 0 & 0 \\
            0 & 0 & 0 & 0 \\
            0 & 0 & 0 & 0 \\
            0 & 0 & 0 & 0 
\end{array}\right) \]
projects $H$ down to a real Hilbert space $H_\R$ whose real dimension
matches the quaternionic dimension of $H$.  By applying
arbitrary $4 \times 4$ real matrices to this subspace we obtain
all of $H$, so $\gamma_* H_\R$ is naturally isomorphic to $H$.   
\end{proof}

Finally, it is worth noting that some of the six functors we have described
have additional nice properties:
\begin{itemize}
\item
The categories $\Hilb_\R$ and $\Hilb_\C$ 
are `symmetric monoidal categories', meaning roughly that they have 
well-behaved tensor products.   The complexification functor 
$\alpha^* \maps \Hilb_\R \to \Hilb_\C$ is a `symmetric monoidal functor',
meaning roughly that it preserves tensor products. 
\item
The categories $\Hilb_\R, \Hilb_\C$ and $\Hilb_\H$ are `dagger-categories',
meaning roughly that any morphism $T \maps H \to H'$ has a Hilbert space
adjoint $T^\dagger \maps H' \to H$ such that
\[          \langle T v, w \rangle = \langle v, T^\dagger w \rangle \]
for all $v \in H$, $w \in H'$.  All six functors preserve this dagger
operation.  
\item
For $\Hilb_\R$ and $\Hilb_\C$, the dagger structure 
interacts nicely with the tensor product, making these categories into 
`dagger-compact categories', and the functor $\alpha^*$ is compatible 
with this as well.  
\end{itemize}
\noindent For precise definitions of the quoted terms here, see
our review articles~\cite{BL,BS}.   For more on
dagger-compact categories see also the work of Abramsky and 
Coecke~\cite{A,AC} (who called them `strongly compact closed'), 
Selinger~\cite{Selinger}, and the book {\it New Stuctures for 
Physics}~\cite{Coecke}.
The three-fold way is best appreciated with the help of these 
category-theoretic ideas, but we have deliberately downplayed them
in this paper, to reduce the prerequisites.  A more 
category-theoretic treatment of symplectic and orthogonal structures 
can be found in our old paper on 2-Hilbert spaces~\cite{HDA2}.  

\section{Solutions}
\label{solutions}

In Section~\ref{problems} we raised two `problems' with 
real and quaternionic quantum theory.  Let us conclude by saying
how the three-fold way solves these.  

We noted that on a real, complex or quaternionic Hilbert space, any 
continuous one-parameter unitary group has a skew-adjoint generator $S$.
In the complex case we can write $S$ as $i$ times a self-adjoint operator
$A$, which in physics describes a real-valued observable.  The first 
`problem' is that in the real or quaternionic cases we cannot do this.  

However, now we know from Theorems \ref{quaternionic_as_complex} and
\ref{real_as_complex} that real and quaternionic Hilbert spaces can 
be seen as a complex Hilbert space with extra structure.   This solves
the problem.  Indeed, we have faithful functors
\[        \alpha_* \maps  \Hilb_\R \to \Hilb_\C   \]
and 
\[        \beta^*  \maps  \Hilb_\H \to \Hilb_\C  . \] 
As noted in the previous section, these are `dagger-functors', so
they send skew-adjoint operators to skew-adjoint operators.

So, given a skew-adjoint operator $S$ on a real Hilbert space $H$, we
can apply the functor $\alpha_*$ to reinterpret it as a skew-adjoint 
operator on the complexification of $H$, namely $H \otimes {{}_\R\C_{\C}}$.
By abuse of notation let us still call the resulting operator $S$.  
Now we can write $S = iA$.  But the resulting self-adjoint operator $A$ 
has an interesting feature: its spectrum is symmetric about $0$!

In the finite-dimensional case, all this means is that for any 
eigenvector of $A$ with eigenvalue $c \in \R$, there is corresponding 
eigenvector with eigenvalue $-c$.  This is easy to see.  Suppose 
that $A v = c v$.  By Theorem \ref{real_as_complex}, the complexification 
$H \otimes {{}_\R\C_{\C}}$ comes equipped with an antiunitary 
operator $J$ with $J^2 = 1$, and we have $SJ = JS$.   It follows that
$J v$ is an eigenvector of $A$ with eigenvalue $-c$:
\[       A J v = - i S J v = - i J S v = J i S v = - J A v = - c J v . \] 
(In this calculation we have reverted to the standard practice of
treating a complex Hilbert space as a {\it left} $\C$-module.)
In the infinite-dimensional case, we can make an analogous but more 
subtle statement about the continuous spectrum. 

Similarly, given a skew-adjoint operator $S$ on a quaternionic
Hilbert space $H$, we can apply the functor $\beta^*$ to reinterpret 
it as a skew-adjoint operator on the underlying complex Hilbert
space.   Let us again call the resulting operator $S$.  We again
can write $S = iA$.  And again, the spectrum of $A$ is symmetric about 
$0$.  The proof is the same as in the real case: now, by Theorem
\ref{quaternionic_as_complex}, the underlying complex Hilbert space
is equipped with an antiunitary operator $J$ with $J^2 = -1$, but
we again have $SJ = JS$, so the same calculation goes through.

The second `problem' is that we cannot take the tensor product of 
two quaternionic Hilbert spaces and get another quaternionic Hilbert
spaces.  But Theorem \ref{tensor} makes this seem like an odd thing
to want!  Just as two fermions make a boson, the tensor product of 
two quaternionic Hilbert spaces is naturally a {\it real} Hilbert space.
After all, Theorem \ref{quaternionic_as_complex} says that a 
quaternionic Hilbert space can be identified with a complex Hilbert space
with an antiunitary $J$ such that $J^2 = -1$.  If we
tensor two such spaces, we get a complex Hilbert space equipped with 
an antiunitary $J$ such that that $J^2 = 1$.   Theorem \ref{real_as_complex}
then says that this can be identified with a real Hilbert space.  

Further arguments of this sort give four tensor product functors:
\[
\begin{array}{ccl}
\tensor \maps \Hilb_\R \times \Hilb_\R & \to & \Hilb_\R \\
\tensor \maps \Hilb_\R \times \Hilb_\H & \to & \Hilb_\H \\
\tensor \maps \Hilb_\H \times \Hilb_\R & \to & \Hilb_\H \\
\tensor \maps \Hilb_\H \times \Hilb_\H & \to & \Hilb_\R .
\end{array}
\]
We can also tensor a real or quaternionic Hilbert space with a 
complex one and get a complex one:
\[
\begin{array}{ccl}
\tensor \maps \Hilb_\C \times \Hilb_\R & \to & \Hilb_\C \\
\tensor \maps \Hilb_\R \times \Hilb_\C & \to & \Hilb_\C \\
\tensor \maps \Hilb_\C \times \Hilb_\H & \to & \Hilb_\C \\
\tensor \maps \Hilb_\H \times \Hilb_\C & \to & \Hilb_\C  .
\end{array}
\]
Finally, of course we have:
\[ 
\begin{array}{ccl}
\tensor \maps \Hilb_\C \times \Hilb_\C & \to & \Hilb_\C .
\end{array}
\]
In short, the `multiplication' table of real, complex and quaternionic
Hilbert spaces matches the usual multiplication table for the numbers
$+1, 0, -1$.   This should remind us of the Frobenius--Schur
indictator, mentioned in the Introduction.  The moral, then, is not 
to fight the patterns of mathematics, but to notice them and follow them.

\subsection{Acknowledgements}  
I thank the organizers of QPL7 for inviting
me to present some of this material in Oxford on May 27th, 2010.  
I thank Torsten Ekedahl, Harald Hanche-Olsen,
David Roberts and others on MathOverflow for
helping me with some math questions, 
Urs Schreiber for reminding me about the 
state-operator correspondence, and Greg Friedman for pointing
out many mistakes.  I also thank the referees for their detailed
suggestions.

\end{document}